\documentclass[final]{siamltex}
\usepackage{amsmath}
\usepackage{mathrsfs}
\usepackage{amsfonts}
\usepackage[dvips]{graphicx}
\usepackage{psfrag}
\usepackage{amsopn}
\usepackage{amstext}
\usepackage{amssymb}
\usepackage{amscd}
\usepackage{latexsym}

\newtheorem{example}[theorem]{Example}

\DeclareMathOperator{\T}{T}
\DeclareMathOperator{\cone}{cone}
\DeclareMathOperator{\I}{I}

\DeclareMathOperator{\trace}{trace}
\DeclareMathOperator{\conv}{conv}
\DeclareMathOperator{\closure}{clo}

\DeclareMathOperator{\expect}{E}
\DeclareMathOperator{\covariance}{Cov}
\DeclareMathOperator{\aff}{aff}

\newcommand{\cU}{\mathscr U}
\newcommand{\cX}{\mathscr X}

\newcommand{\Cov}{Q}
\newcommand{\comment}[1]{}
\newcommand{\N}{\mathbb{N}}
\newcommand{\R}{\mathbb{R}}

\newcommand{\Half}{\ensuremath{\frac{1}{2}}}
\newcommand{\Sym}[1]{\ensuremath{{\mathscr S}^{#1}}}

\begin{document}

\title{Relative Robust Portfolio Optimization}

\author{Raphael Hauser\footnotemark[1], 
Vijay Krishnamurthy\footnotemark[2]\; and 
Reha H.\ T{\"u}t{\"u}nc{\"u}\footnotemark[3]}

\date{\today}

\renewcommand{\thefootnote}{\fnsymbol{footnote}}
\footnotetext[1]{(Corresponding author) 
Mathematical Institute, University of Oxford, 24–-29 St Giles', Oxford OX1 3LB, 
United Kingdom. hauser@maths.ox.ac.uk. This author
was supported through grants GR/S34472 and EP/H02686X/1 from the
Engineering and Physical Sciences Research Council of the UK.}
\footnotetext[2]{LL Funds LLC, Philadelphia, PA, USA. vijay.krishnamurthy@llfunds.com}
\footnotetext[3]{Goldman Sachs Asset Management, New York, NY, 
USA. reha.tutuncu@gs.com.}
\renewcommand{\thefootnote}{\arabic{footnote}}

\maketitle

\begin{abstract}
Considering mean-variance portfolio problems with uncertain model parameters, we contrast the 
classical absolute robust optimization approach with the relative robust approach based on a maximum 
regret function. Although the latter problems are NP-hard in general, we show that tractable inner and 
outer approximations exist in several cases that are of central interest in asset management. 
\end{abstract}



\section{Introduction} 

We consider decision-making tools for problems that are defined by
uncertain or unknown parameters. Uncertainty and the risk of 
undesirable outcomes are inevitable features of most 
productive activities. From engineers to economists, from health care 
providers to investment managers, many professionals must take
decisions under considerable uncertainty on a daily basis.
Our objective is to improve the quality of decisions made 
in these environments by understanding, modeling, quantifying 
and managing uncertainty. To achieve this goal we propose a robust optimization modeling 
methodology based on extending the idea of relative robust optimization of Kouvelis and Yu 
\cite{KY} in the context of discrete optimization. 

Many decision problems with uncertainty 
can be formulated as optimization problems.
In recent years, robust optimization (RO) has emerged as a 
powerful tool for managing uncertainty in such optimization problems
\cite{BTN1, BTN2}. An excellent overview can be gained from the recent 
survey paper \cite{bertsimas}. 
Robust optimization is a generic term that is 
used to describe a class of modeling strategies as well as 
solution methods for optimization problems that are
defined by uncertain inputs. Decisions made with 
incomplete information may result in undesirable outcomes when 
the realized values of the uncertain inputs are unfavorable.
Robust optimization models and algorithms aim to mitigate the effects 
of uncertainty and obtain a solution that is guaranteed to 
perform reasonably well for all, or at least most, 
possible realizations of the uncertain input parameters.

There is growing evidence, 
both empirical and theoretical, that robust-optimized 
solutions have better characteristics than 
their non-robust counterparts. For example, a theoretical study by
Sch\"ottle and Werner demonstrates that the map between the model parameters 
of an optimization problem and its set of optimal solutions can become much
smoother if one uses robust optimization with ellipsoidal uncertainty
sets \cite{Schottle}, \cite{Schottle2}. In a portfolio optimization setting with uncertain
expected return estimates, Ceria and Stubbs report simulated results
where the ex-post performance of robust-optimal portfolios outperform 
those of standard mean-variance optimal portfolios with high 
frequency \cite{Ceria_Stubbs}. 

While the stated goals of robust optimization are intuitive, it is
not always clear what metrics one should use to achieve these goals.
Most robust optimization approaches discussed in the existing literature 
use the ``worst-case objective value'' as the comparison metric among alternative sets of decisions.
Despite the advantages we mentioned in the previous paragraph, the focus on the worst-case 
objective value in robust optimization is a source of frequent criticism. Modelers, often with 
good reason, worry that the extreme scenarios in the uncertainty 
set have an undue influence on the final decisions in such robust
formulations. The ``worst-case objective value'' is an {\em absolute} metric. 
While there are many situations where it is the appropriate metric 
for evaluating robustness, it is inadequate for measuring robustness in a {\em relative} sense. 

A typical example arises in the investment management context 
where managers are frequently evaluated and compensated based 
on their performance relative to the competition. 
For robust decision-making in an uncertain decision environment, 
rather than protecting themselves against worst-case scenarios, 
investment managers may thus prefer to choose decisions that avoid 
falling severely behind their competitors under a range of scenarios. This view of robustness was 
formalized by Kouvelis and Yu \cite{KY}. For each choice of the decision variables and each scenario 
one compares the attained objective value with the optimal objective value attainable under 
the model parameter values described by the scenario. The 
difference between these two values, or alternatively their ratio,
can be seen as measures of {\em regret} based on hindsight after the true
values of the uncertain parameters are revealed. 
With the objective of limiting such regret measures
Kouvelis and Yu arrive at the {\em robust deviation} and 
{\em relative robust} decision problem formulations. 
We will provide these formulations in Section \ref{sec:ro}
as they form the focus of our study. We will refer to 
the robust deviation and relative robust decision problems collectively
as relative robust problems. 

While Kouvelis and Yu explore robust deviation and relative robust
decisions in several classes of discrete optimization problems,
similar studies for continuous optimization problems are 
mostly missing in the literature. A rare exception is 
Taguchi's master's thesis \cite{Taguchi} and the subsequent paper \cite{Taguchi2}. 
As these authors also observed,
relative robust formulations are typically more difficult than the
corresponding absolute robust formulations. 
Since they involve the optimal value function whose argument is the
vector of uncertain parameters inside a min-max optimization 
problem, relative robust problems are three-level optimization problems.
This is in contrast to the two-level absolute robust formulations. 
Since the optimal value function is rarely available in closed form, 
tractability is an important concern
for these models. Our study shows that for many uncertainty structures
on quadratic programming and other optimization problems, the resulting
relative robust formulation can be reduced to one or a series of 
single-level deterministic optimization problems that can be solved
using conic optimization methods. 

The simplest uncertainty sets are finite sets, corresponding to the
intuitive notion of a collection of scenarios. Both the absolute
and relative robust formulations with finite uncertainty sets are 
relatively easy as they can be solved as a finite sequence of deterministic
problems. Using simple convexity arguments, we show that robust
problems with polytopic uncertainty structures (uncertainty sets 
defined as convex hull of a finite number of points) can be reduced to the 
finite case and are therefore of the same complexity. Relative robust 
optimization problems with polytopic uncertainty sets were also considered 
by Taguchi et.\ al.\ \cite{Taguchi2}, but our discussion is based on the 
second author's MSc thesis \cite{krishnamurthy} which predated the 
work of the aforementioned authors and formed the basis of the original 
draft of this paper. Gregory et.\ al.\ \cite{gregory} also investigated the 
polyhedral case. 

In the second part of the paper we move beyond polytopic uncertainty sets 
and consider relative robust models with ellipsoidal uncertainty. In the application 
on which we concentrate, mean-variance portfolio optimization, we assume 
ellipsoidal uncertainty only in the vector of expected returns and assume that 
the covariance matrix is fixed. This assumption is justified, as practitioners 
typically use matrix shrinkage and factor models to ensure that the covariance 
is estimated robustly. Ellipsoidal uncertainty sets for the vector of expected 
returns appear quite naturally as confidence regions in their statistical estimation. 
Taguchi et.\ al.\ \cite{Taguchi2} propose to approximate ellipsoidal uncertainty 
sets by a polytope obtained as the convex hull of a random sample of points 
from the ellipsoid. The resulting relative robust problems are relaxations of the 
original relative robust problem with ellipsoidal uncertainty, and since these 
relaxations are based on polyhedral uncertainty, they can be solved via the 
approach discussed earlier. Adom \cite{Adom} investigated similar relaxations 
but based on polyhedra generated by pseudo-randomly chosen points, which 
leads to faster convergence. 

Our approach to relative robustness under ellipsoidal uncertainty is very different. 
By developing inner approximations to the relative robust problem in the form of 
a symmetric cone programming problem, we obtain solutions that are guaranteed 
to be feasible. In several cases of interest our inner approximations are 
provably tight, as we show using a theory developed by Sturm and Zhang \cite{sturm}. 
Thus, while our models approximate  the feasible set of the relative robust problem 
from the inside, Taguchi et.\ al.'s approach \cite{Taguchi2} approximates it from the 
outside, which yield solutions that are not guaranteed to be feasible but give a 
bound on the optimal objective value. The two approaches can be combined to 
obtain approximation guarantees for both -- a distance to optimality in the 
case of our model, and a distance to feasibility in the case of the model of 
Taguchi et.\ al..

Since we have in mind the mean-variance portfolio optimization framework of 
Markowitz \cite{Markowitz} as a particular application of the methods 
investigated in this paper, most of the notation we use will be inspired by 
this framework which is explained in some detail in Section \ref{application}. 
Much of the other notation is of standard use in the optimization literature, 
such as $\succeq$ and $\succ$ to denote positive semidefiniteness and 
definiteness of a matrix, for example, $\cone(\cdot)$, $\conv(\cdot)$ 
and $\aff(\cdot)$ for conic, convex and affine hulls of a subset of a 
vector space respectively, $\bullet$ for the trace inner-product of two 
matrices of equal size, $\cdot^*$ for duals of functionals and cones, 
$e$ and $\I$ for the vector of all ones and the identity matrix of appropriate 
size, and $\cdot^{\T}$ for the transpose of a matrix. 

\section{Absolute Robust versus Relative Robust Optimization} \label{sec:ro}

We consider a generic optimization problem
whose input parameters are denoted by the vector $p$.
\begin{align}
\max_{x\in\R^n}\,&f(x,p)\label{eq:op}\\
\text{s.t.}\quad&x\in\cX_p,\nonumber
\end{align}
where $f(x,p)$ and $\cX_p$ represent the objective function and 
the feasible set of the problem. For a given $p$, let $z^*(p)$ and 
$\Omega^*(p)$ denote, respectively, the optimal value and the
set of optimal solutions of the above given problem, provided that they exist. 
We will also use the notation $x^*(p)$ to denote a generic element
of $\Omega^*(p)$.

\subsection{The Absolute Robust Optimization Framework}

When $p$ is known \eqref{eq:op} is a standard optimization
problem. Robust optimization (RO) is concerned with the
case where $p$ is not known with certainty. In recent years, 
RO has emerged as an alternative to 
traditional approaches to optimization under uncertainty
such as sensitivity analysis and stochastic programming. 
As mentioned in the introduction, its primary objective is to find 
solutions that will have a good performance under a 
variety of scenarios for the uncertain input parameters.
RO models are especially well-suited in situations
where there are constraints with uncertain parameters that must
be satisfied regardless of the values of these parameters or 
when the optimal solutions are 
particularly sensitive to perturbations. Additionally, RO is an 
attractive modeling option when the decision-maker cannot afford 
low-probability high-magnitude risks. 

One of the essential elements of a RO model is the {\em uncertainty
set}. The uncertainty set, say $\cU$, represents the set
of possible scenarios/realizations for the parameters $p$.
When $p$ is uncertain and must be estimated, uncertainty sets can 
represent or be formed by difference of opinions, alternative estimates, 
confidence regions of statistical estimators, or based on Bayesian or Kalman 
filtering methods for tracking the evolution of an assumed probability 
distribution for $p$.  While the current literature does not 
provide clear guidelines on their construction, 
their shape often reflects the sources of uncertainty while
their size depends on the desired level of robustness. Common
types of uncertainty sets include: 
(i) $\cU = \{ p_1, p_2, \ldots, p_k \}$ (a finite set
of scenarios), (ii) $\cU = \mbox{conv} (p_1, p_2, \ldots, p_k )$ 
(a polytopic set), (iii) $\cU = \{p: l \leq p \leq u\}$ (intervals), 
and (iv) $\cU = \{p: p=p_0+Mu, \|u\| \leq 1\}$ (an ellipsoidal set).

RO formulations optimize some variation of a 
worst-case performance metric, where the ``worst-case'' is computed over the
uncertainty set. In most cases \cite{Ceria_Stubbs, 
ElGhaoui_Oks_Oustry, Goldfarb_Iyengar, Halldorsson_Tutuncu},
the objective is to optimize the worst-case realization of the objective
function. For the optimization problem \eqref{eq:op}, this leads
to the following formulation:
\begin{equation}\label{eq:absrob}
 \max_{x \in \bigcap_{p \in\cU}\cX_p}\left(\min_{p \in \cU}f(x,p)\right).
\end{equation}

Kouvelis and Yu \cite{KY} classify \eqref{eq:absrob} as 
the {\em absolute robust decision problem}. 
This name reflects the fact that the worst-case
objective value is an absolute metric. One potential consequence
of this emphasis on the worst-case is that the decisions are 
disproportionately affected by extreme scenarios in the uncertainty
set. As this is not always desirable, Bertsimas and Sim \cite{bertsimasSim} 
study this cost of robustness as a function of the level of conservatism. An alternative we 
consider in this paper is to seek robustness in a relative sense. 

\subsection{The Relative Robust Optimization Framework}

For this purpose, we consider a {\em regret function} that measures the
difference between the performance of the solution with and without the benefit of hindsight.
If we choose $x$ as decision vector when $p$ is the vector of realized 
parameter values, then the {\em regret} associated with having chosen $x$ rather than 
$x^*(p)$ as decision vector is defined as follows, 
\begin{eqnarray} \label{eq:regret}
r(x,p):= z^*(p) - f(x,p) = f(x^*(p),p)- f(x,p).
\end{eqnarray} 
Note that since $x^*(p)$ is an optimal decision vector for the parameter 
values $p$, the regret $r(x,p)$ is always nonnegative. 

The regret function is not useful at the decision-making stage since 
we cannot measure the regret before we observe the realized value of
the parameters. Furthermore, in many a context $p$ cannot be 
observed even {\em after} its realization. For example, financial data 
typically yield a single sample of a random return vector $R$ while the 
parameters $p=(\expect[R],\covariance(R))$ that serve as input 
parameters to the optimal investment problem are neither directly 
observable nor inferable from this single sample. 
For this reason we consider the maximum regret function instead, which
provides an upper bound on the true regret,
\begin{eqnarray} \label{eq:maxreg}
R(x) := \max_{p \in \cU} r(x,p) = \max_{p \in \cU}\left(z^*(p) - f(x,p)\right).
\end{eqnarray} 
If the function $z^*(p)$ is positive everywhere, one can also consider a 
scaled version of the regret function,
\begin{eqnarray} \label{eq:regret2}
 \tilde{r}(x,p)& =&\frac{z^*(p) - f(x,p)}{z^*(p)},\\
\tilde{R}(x) &= &\max_{p \in \cU}\tilde{r}(x,p)=
\max_{p \in \cU} \frac{z^*(p) - f(x,p)}{z^*(p)}.
\end{eqnarray} 

In Kouvelis and Yu \cite{KY}, vectors $x$ that minimize the maximum 
regret functions $R(x)$ and $\tilde{R}(x)$ are called {\em robust 
deviation decisions}
and {\em relative robust decisions} respectively. We 
collectively refer to problems seeking such decisions as relative robust problems
and focus on the function $R(x)$ for most of the rest of our discussion.

Let us consider the simpler case where the uncertain parameters are 
only in the objective function and the feasible set
$\cX_p \equiv \cX$ is independent of $p$. In most models,
the dependence of the objective function on the
uncertain parameters is linear. When this is the case, it is 
easy to see that the optimal value function $z^*(p)$ is a convex function. 
In fact, it is sufficient that $f$ be convex in $p$ to guarantee the  
convexity of $z^*(p)$:
\begin{lemma} \label{lem:simple}
Let $\cU$ be a convex set. 
For all $p \in \cU$, define
\begin{eqnarray*}
z^*(p)&= \sup_{x \in \cX}&f(x,p)
\end{eqnarray*}
where $f$ is convex in $p$. Then $z^*$ is a convex function
on $\cU$.
\end{lemma} 

Lemma \ref{lem:simple} is part of the folklore on convex analysis. For 
the sake of completeness we include a proof in Appendix A. As we will see in the next section, 
this simple convexity result is responsible for reducing relative robust optimization models with 
polytopic uncertainty sets to problems with finite uncertainty sets.

\section{Application to Mean-Variance Portfolio Optimization}\label{application}

Portfolio theory deals with the problem of deciding what proportion of 
investable wealth to allocate to each of several risky investment opportunities 
so as to achieve a chosen goal, which is usually to maximize the expected return 
while limiting risk. Under the mean-variance optimization (MVO) approach of 
Markowitz \cite{Markowitz}, all $n$ investments are assumed to be held during 
the same fixed investment period over which they generate random returns 
$R_i$. Assembled in a random vector $R$, the expectation $\mu=\expect[R]$ 
and the positive definite covariance matrix $\Cov=\covariance(R)$ of the asset 
returns serve as input parameters to 
one of the quadratic programming problems \eqref{eq:minvar}--\eqref{eq:maxrar} 
described below. Their solutions yield optimal portfolio weights. 
Though these problems are convex and computationally 
tractable and can therefore be solved to global optimality, MVO models can produce 
portfolios that are highly sensitive to the values of the model parameters 
$(\mu,\Cov )$ and show unsatisfactory diversification. Since $(\mu,\Cov )$ have 
to be estimated statistically, output sensitivity to these parameters is an 
important practical issue. Robust optimization models have therefore emerged 
as favourable alternatives to plain MVO models \cite{Goldfarb_Iyengar, 
Koenig_Tutuncu, Ceria_Stubbs}. 

\subsection{Classical Mean-Variance Portfolio Models}

Although the mathematical methods investigated in this paper are applicable more widely, the 
MVO framework constitutes their main motivation and application. 
We shall therefore briefly describe some of the models that arise in this context. 
Let $x_i$ be the proportion of wealth invested in the $i$-th investment 
opportunity (or asset), and let these weights be collected in a vector $x$ 
of size $n$. The portfolio corresponding to the weights $x$ then has the 
overall return $R^{\T}x$ with expectation $\mu^{\T}x$ and variance 
$x^{\T}\Cov x$. Apart from the {\em budget constraint} $e^{\T}x=1$ 
(where $e:=[\begin{smallmatrix}1&\dots&1\end{smallmatrix}]^{\T}$), 
fund managers usually restrict the set $\cX$ of {\em feasible} portfolios 
(investment decisions they are willing to consider) by introducing further 
constraints that impose limits on short-selling, diversification, rebalancing 
costs and other criteria. Typically, all constraints are linear, leading to a 
polyhedral feasible set $\cX=\{x\in\R^n:\,Fx=f,\,Gx\leq g\}$. Here we make 
the minimal assumption that $\cX$ be a convex tractable set, that is, a set 
for which it can be decided in polynomial time whether or not a given point 
is a member. 

\subsubsection{Convex MVO Models}\label{subsubsection:variants}

Taking the variance of the portfolio return as a risk measure, MVO 
formulations are obtained by either choosing to minimize the variance 
subject to a lower bound target return $\rho$, to maximize the return 
subject to an upper bound target risk $\sigma^2$ or to maximize the 
risk-adjusted expected return $\mu^{\T}x-\lambda x^{\T}\Cov x$ defined by a 
specific choice of a risk-aversion parameter $\lambda>0$,
\begin{align}
\min_{x\in\R^n}\,&f_{\Cov}(x):=x^{\T}\Cov x \label{eq:minvar}\\
\text{s.t.}\quad&\mu^{\T}x\geq \rho,\nonumber\\
&x\in\cX,\nonumber
\end{align}
\begin{align}
\max_{x\in\R^n}\,&f_{\mu}(x):=\mu^{\T}x\label{eq:maxret}\\
\text{s.t.}\quad&x^{\T}\Cov x\leq\sigma^2,\nonumber\\
&x\in\cX,\nonumber
\end{align}
\begin{align} 
\max_{x\in\R^n}\,&f_{\mu,\Cov}(x):=\mu^{\T}x-\lambda x^{\T}\Cov x
\label{eq:maxrar}\\
\text{s.t.}\quad&x\in\cX.\nonumber
\end{align}
It is well-known and can easily be established using KKT conditions that the 
three formulations presented above are equivalent in the sense that they 
produce identical solutions for 
appropriately chosen values of $\rho$, $\sigma^2$ and 
$\lambda$. For example, there exists a function 
$\sigma^2(\mu,\Cov ,r)$ such that the solutions of \eqref{eq:minvar} and
\eqref{eq:maxret} coincide when $\sigma^2=\sigma^2(\mu,\Cov ,\rho)$. 
We note however that this functional dependence also depends on the 
model parameters. A portfolio $x$ is said to be {\em efficient} if it optimizes 
\eqref{eq:minvar}--\eqref{eq:maxrar} for some choice of $\rho$, 
$\sigma^2$ and $\lambda$ respectively. Because of the 
above-mentioned equivalence, it does not matter which problem we use for 
this definition. The set $\{((x^{\T}\Cov x)^{1/2},\mu^{\T}x):\,x\text{ efficient}\}$ 
is the efficient frontier.

\subsubsection{Sharpe Ratio Maximization}\label{subsubsection:Sharpe}

Another variant of MVO is the Sharpe Ratio maximization problem. 
Let $r $ be the return of a risk-free investment held over the same period 
as the above considered assets, e.g., a short-term government bond or the 
money market. If this risk-free asset is included among the considered assets 
-- we call it asset $0$ and refer to it as cash -- and if the feasible set $\cX_0$ 
of portfolios containing a cash position is of affine form, 
\begin{equation}\label{affine form}
\cX_0=\aff\left(\left\{\bigl[\begin{smallmatrix}0\\x\end{smallmatrix}\bigr]:\,
x\in\cX\right\}\cup\left\{\bigl[\begin{smallmatrix}1\\0\end{smallmatrix}\bigr]\right\}\right),
\end{equation}
where $\aff(\cdot)$ denotes the affine hull of a set and $\cX$ is the set of 
feasible portfolios containing only positions in the risky assets, then cash 
can freely be borrowed to invest in the risky assets. In this case the 
efficient frontier is a straight line going through the point $(0,r )$ and 
with a gradient given by 
\begin{align}
\max_{x\in\R^n}\,&\frac{\mu^{\T}x - r }{\sqrt{x^{\T} \Cov x}}\label{eq:Sharpe} \\
\text{s.t.}\quad&x\in\cX.\nonumber
\end{align}
The {\em Sharpe ratio} \cite{Sharpe} of a portfolio $x$ is defined as the 
ratio of the excess expected return of the portfolio over the risk-free asset 
and the standard deviation of the portfolio return. Correspondingly, Model 
\eqref{eq:Sharpe} is called the 
{\em Maximum Sharpe Ratio problem} (MSR). In the case where this problem has a 
unique optimal solution, this solution is called the {\em market portfolio}. It can be easily 
seen that any efficient portfolio is then an affine combination of cash and the market portfolio. 
This observation forms the basis of Sharpe's capital asset pricing theory \cite{capm}. 

While \eqref{eq:Sharpe} is a nonlinear and nonconvex problem, it can be solved by the 
tractable convex programming problem 
\begin{align}
\max_{y\in\R^{n}}\,&g(y)=(\mu-r e)^{\T}y\label{convexification}\\
\text{s.t.}\quad&y\in\R_+\cX,\nonumber\\
&y^{\T}\Cov y\leq 1.\nonumber
\end{align}
See Appendix B for detailed explanations, and also 
\cite{Goldfarb_Iyengar} and \cite{Koenig_Tutuncu} for similar techniques. 

\subsection{Absolute Robust Portfolio Models}\label{subsec:arpm}

Motivated by the
sensitivity of the solutions of Problems \eqref{eq:minvar}, \eqref{eq:maxret}, 
\eqref{eq:maxrar} and \eqref{eq:Sharpe} as functions of the model parameters 
$(\mu,\Cov)$, we next consider 
robust counterparts of these models. Depending on how the uncertainty 
set $\cU$ for the model parameters $(\mu,\Cov)$ is chosen and which of 
the models one chooses to robustify, one 
arrives at different robust formulations. We note that the ensuing 
models are no longer all equivalent in the sense in which the nonrobust 
versions \eqref{eq:minvar}--\eqref{eq:maxrar} were. The main reason 
for this difference is that the feasible 
sets of Problems \eqref{eq:minvar} and \eqref{eq:maxret} depend on 
the model parameters $(\mu,\Cov)$ while those of Problems \eqref{eq:maxrar} 
and \eqref{eq:Sharpe} do not, with the consequence that in the context of 
the former two problems joint uncertainty structures in $\mu$ and $\Cov$ 
cannot be exploited, while in the context of the latter two they can, at least conceptually. 

To be more specific, let us write $\cU_{\mu}:=\left\{\mu:\,\exists\,(\mu,\Cov)\in\cU
\right\}$ for the projection of $\cU$ onto the $\mu$-component, and 
$\cU_{\Cov}:=\left\{\Cov:\,\exists\,(\mu,\Cov)\in\cU\right\}$ 
for its projection onto the $\Cov$-component, and note that when 
$\mu$ and $\Cov$ have joint uncertainty structure, then $\cU$ is usually 
strictly contained in the Cartesian product $\cU_{\mu}\times\cU_{\Cov}$. 
Yet the absolute robust counterpart of \eqref{eq:minvar} 
has the following equivalent formulations, 
\begin{align}
\min_{\{x\in\cX:\,\mu^{\T}x\geq \rho\,\forall\,(\mu,\Cov)\in\cU\}}
\,&\left(\max_{(\mu,\Cov)\in\cU}x^{\T}\Cov x\right)\nonumber\\
\Leftrightarrow 
\min_{\{x\in\cX:\,\min_{\mu\in\cU_{\mu}}\mu^{\T}x\geq \rho\}}
\,&\left(\max_{\Cov\in\cU_{\Cov}}x^{\T}\Cov x\right)\label{eq:ARminvar},\\
\Leftrightarrow 
\min_{\{x\in\cX:\,\mu^{\T}x\geq \rho\,\forall\,(\mu,\Cov)\in
\cU_{\mu}\times\cU_{\Cov}\}}
\,&\left(\max_{(\mu,\Cov)\in\cU_{\mu}\times\cU_{\Cov}}x^{\T}\Cov x\right).\nonumber
\end{align}
Thus, the joint uncertainty structure of $\mu$ and $\Cov$ cannot be exploited 
in the context of the robust problem \eqref{eq:ARminvar}, and neither can it 
be in the context of the absolute robust counterpart of \eqref{eq:maxret}, 
\begin{equation}
\max_{\{x\in\cX:\,\max_{\Cov\in\cU_{\Cov}}x^{\T}\Cov x\leq\sigma^2\}}
\,\left(\min_{\mu\in\cU_{\mu}}\mu^{\T}x\right).\label{eq:ARmaxret}
\end{equation}
In contrast, joint uncertainty in $\mu$ and $\Cov$ can be exploited, at least 
conceptually, in the framework of the absolute robust counterpart of 
\eqref{eq:maxrar}, 
\begin{equation}\label{eq:ARmaxrar}
\max_{x\in\cX}\left(\min_{(\mu,\Cov)\in\cU}\mu^{\T}x-\lambda x^{\T}\Cov x\right),
\end{equation}
as well as in the absolute robust counterpart of \eqref{eq:Sharpe}, 
\begin{equation}\label{eq:ARSharpe}
\max_{x\in\cX}\left(\min_{(\mu,\Cov)\in\cU}
\frac{\mu^{\T}x-r }{\sqrt{x^{\T}\Cov x}}\right).
\end{equation}

Special cases of the above-described models appear in the literature as 
follows: Goldfarb and Iyengar \cite{Goldfarb_Iyengar} discussed the problems 
\eqref{eq:ARminvar}, \eqref{eq:ARmaxret} in the 
case where $\cU$ is an uncertainty set of Cartesian type 
$\cU=\cU_{\mu}\times\cU_{\Cov}$ corresponding to a confidence region 
for the statistical estimators arising in the context of the load-factor model 
for $(\mu,\Cov)$. Halld\'{o}rsson and T\"{u}t\"{u}nc\"{u} \cite{Halldorsson_Tutuncu}
discussed the problem \eqref{eq:ARmaxrar} in the case where 
$\cU=\cU_{\mu}\times\cU_{\Cov}$, and where $\cU_{\mu}$ is a box of confidence 
intervals for the individual components of $\mu$ and $\cU_{\Cov}$ is a box of 
confidence intervals intersected with the cone of positive semidefinite symmetric 
matrices. 

Note that the absolute robust problems \eqref{eq:minvar}--\eqref{eq:maxrar} 
are all two-level optimization problems and thus a priori harder to solve than the classical 
nonrobust models \eqref{eq:minvar}--\eqref{eq:maxrar}. However, when 
$\cU$ is tractable, then the robust problems are tractable too. See the 
above cited papers and the other literature on robust optimization for details. 

\subsection{Relative Robust Portfolio Models}

To contrast the relative robust framework with the absolute robust setting, 
we next formulate relative robust counterparts of problems 
\eqref{eq:minvar}--\eqref{eq:maxrar}. 

Let us first consider problem \eqref{eq:minvar}, which has two meaningful 
relative robust analogues: In the first version, 
\begin{equation*}
z^*(\mu,\Cov ):=\min_{\{y\in\cX:\,\mu^{\T}y\geq \rho\}}\,y^{\T}\Cov y,
\end{equation*}
is defined as the maximum objective value achievable by an 
{\em omniscient} 
adversary (one who knows the parameters $(\mu,\Cov )$ with certitude), 
leading to the maximum regret 
\begin{equation*}
R(x):=\max_{(\mu,\Cov)\in\cU}(x^{\T}\Cov x\quad-\min_{\{y\in\cX:\,
\mu^{\T}y\geq \rho\}}y^{\T}\Cov y)
\end{equation*}
and the relative robust problem 
\begin{equation}\label{relrob:minvar version1}
\min_{\{x\in\cX:\,\min_{\mu\in\cU_{\mu}}\mu^{\T}x\geq \rho\}}
\bigl(\max_{(\mu,\Cov)\in\cU}(x^{\T}\Cov x\quad-\min_{\{y\in\cX:\,
\mu^{\T}y\geq \rho\}}y^{\T}\Cov y)\bigr).
\end{equation}
In the second version regrets are computed merely relative to the solution of 
a {\em fortuitous} adversary who is bound to choosing a portfolio that is feasible 
for all parameter values in $\cU$ but happens to choose the one that is 
optimal among these for the true parameter values. Thus, one would have to define 
\begin{eqnarray*}
z^*(\mu,\Cov):=\min_{\{y\in\cX:\,\min_{\nu\in\cU_{\mu}}
\nu^{\T}y\geq \rho\}}\,y^{\T}\Cov y,\\
R(x):=\max_{\Cov\in\cU_{\Cov}}(x^{\T}\Cov x\quad-\min_{\{y\in\cX:\,
\min_{\nu\in\cU_{\mu}}\nu^{\T}y\geq \rho\}}y^{\T}\Cov y),
\end{eqnarray*}
which leads to the relative robust problem 
\begin{equation}\label{relrob:minvar version2}
\hspace{-0.1cm}
\min_{\{x\in\cX:\,\min_{\mu\in\cU_{\mu}}\mu^{\T}x\geq \rho\}}\bigl(
\max_{\Cov\in\cU_{\Cov}}(x^{\T}\Cov x\quad-\min_{\{y\in\cX:\,
\min_{\nu\in\cU_{\mu}}\nu^{\T}y\geq \rho\}}y^{\T}\Cov y)
\bigr).
\end{equation}
Note that, similarly to what we observed in the context of Problem 
\eqref{eq:ARminvar}, the dependence of the feasible set of Problem \eqref{eq:minvar} 
on the model parameter $\mu$ introduces limitations on the exploitation of 
joint uncertainty in $\mu$ and $\Cov$. The emergence of two conceptually different relative 
robust analogues is also due to this dependence. 

In complete similarity, Problem \eqref{eq:maxret} has the following two relative 
robust analogues with similar limitations on the exploitation of structured 
uncertainty sets, 
\begin{align}
\min_{\{x\in\cX:\,\max_{C\in\cU_{\Cov}}x^{\T}C x\leq\sigma^2\}}&\bigl(
\max_{(\mu,\Cov)\in\cU}(\max_{\{y\in\cX:\,
y^{\T}\Cov y\geq\sigma^2\}}\mu^{\T}(y-x))\bigr),\label{relrob:maxret version1}\\
\min_{\{x\in\cX:\,\max_{C\in\cU_{\Cov}}x^{\T}C x\leq\sigma^2\}}&
\bigl(\max_{\mu\in\cU_{\mu}}(\max_{\{y\in\cX:\,\max_{\Cov\in\cU_{\Cov}}
y^{\T}\Cov y\geq\sigma^2\}}\mu^{\T}(y-x))\bigr).
\label{relrob:maxret version2}
\end{align}

Let us now turn our attention to the relative robust counterparts of \eqref{eq:maxrar} 
and \eqref{eq:Sharpe}. In these cases, the feasible set is independent of the model parameters, and 
for any given uncertainty structure there exists only one relative robust counterpart model. 

\subsubsection{The Relative Robust Counterpart of Problem (\ref{eq:maxrar})}

We define 
\begin{align*}
z^*(\mu,\Cov)&:=\max_{y\in\cX}(\mu^{\T}y-\lambda y^{\T}\Cov y),\\
R(x)&:=\max_{(\mu,\Cov)\in\cU}(
\max_{y\in\cX}(\mu^{\T}y-\lambda y^{\T}\Cov y)
\;-\;(\mu^{\T}x-x^{\T}\Cov x)).
\end{align*}
Problem \eqref{eq:maxrar} then has the following relative robust counterpart,  
\begin{equation}\label{eq:rrp1}
\min_{x\in\cX}\bigl(\max_{(\mu,\Cov)\in\cU}(
\max_{y\in\cX}(\mu^{\T}y-\lambda y^{\T}\Cov y)
\;-\;(\mu^{\T}x-x^{\T}\Cov x))\bigr).
\end{equation}
Introducing an artificial variable $\gamma$ that expresses an upper 
bound on the regret, we can equivalently reformulate this problem as follows, 
\begin{align}
\min_{(x,\gamma)\in\R^{n+1}}\,&\gamma\label{main form I}\\
\text{s.t.}\quad&x\in\cX,\nonumber\\
&\gamma\geq z^*(\mu,\Cov)-\mu^{\T}x+x^{\T}\Cov x,\quad\forall\,
(\mu,\Cov)\in\cU.\nonumber
\end{align}
Alternatively, using the definition of $z^*(\mu,Q)$, we obtain another equivalent 
formulation, 
\begin{align}
\min_{(x,y,\gamma)\in\R^{2n+1}}\,&\gamma\label{main form II}\\
\text{s.t.}\quad&x\in\cX,\nonumber\\
&\gamma\geq \mu^{\T}y-y^{\T}\Cov y-\mu^{\T}x+x^{\T}\Cov x,\quad\forall\,
(\mu,\Cov)\in\cU, y\in\cX.\nonumber
\end{align}
We will use the formulation \eqref{main form I} in Section \ref{sec:tope}, 
while the formulation \eqref{main form II} will be preferable in Section 
\ref{sec:ellipsoidal}. 

\subsubsection{The Relative Robust Counterpart of Problem (\ref{eq:Sharpe})}

We now describe the relative robust maximum Sharpe-ratio problem and reformulate it using 
the convexification approach described in Appendix B. 

We will use the notation introduced in 
Section \ref{subsubsection:Sharpe} and Appendix B, and we assume 
that the feasible set $\cX_0$ of portfolios containing a cash position takes the affine form 
\eqref{affine form}.

For any given $(\mu,\Cov)\in\cU$ the maximum Sharpe Ratio achievable in $\cX$ is given by 
\begin{equation}\label{eq:frac}
 z^*(\mu,\Cov):=\max_{x\in\cX}\frac{(\mu-r  e)^{\T}x}
{\sqrt{x^{\T}\Cov x}}
\end{equation}
and can be computed by solving a convex problem of the form \eqref{convexification}. 

Furthermore, introducing an artificial variable $\gamma$, the relative robust counterpart of Problem 
\eqref{eq:Sharpe} can be formulated as follows, 
\begin{align*}
\min_{\gamma\in\R, x\in\cX}\,&\gamma\\
\text{s.t. }&\frac{(\mu -r  e)^{\T}x}{\sqrt{x^{\T}\Cov x}}\geq z^*(\mu,\Cov) -\gamma,\quad
\forall\,(\mu,\Cov)\in\cU.
\end{align*}
We see that, in effect, this model correponds to comparing the Sharpe ratio achieved by the 
optimal decisions $x^*$ with the Sharpe Ratio achieved by an omniscient adversary. Although 
the absolute robust counterpart problem is tractable when the uncertainty set is of cartesian form 
$\cU=\cU_{\mu}\times\cU_{\Cov}$, the relative robust counterpart problem is not. We therefore 
restrict ourselves to the case where only $\mu$ is uncertain, and $\Cov$ is known with certainty, 
that is, $\cU=\cU_{\mu}\times\{\Cov\}$. This is a realistic assumption, as in practical applications it 
is typically more interesting to model uncertainty in $\mu$ only and guard against uncertainty in 
$\Cov$ via matrix shrinkage techniques. 

Assuming an uncertainty structure of the form $\cU=\cU_{\mu}\times\{\Cov\}$ and using the 
convexification approach described in Appendix B, the above relative robust model is equivalent to  
\begin{align}
\min_{\gamma\in\R, y\in\R^n}\,&\gamma\label{eq:RRSharpe}\\
\text{s.t. }&(\mu -r  e)^{\T}y\geq z^*(\mu) -\gamma,\quad\forall\,\mu\in\cU_{\mu},\nonumber\\
&y\in\R_+\cX,\nonumber\\
&y^{\T}\Cov y\leq 1.\nonumber
\end{align}
Note that since $\Cov$ is certain, $z^*$ can be considered to be a function of $\mu$ only. We also 
remark that if we had considered the relative robust counterpart of Problem \eqref{convexification}, 
we would have arrived at the same model under the chosen uncertainty structure.

\subsubsection{Complexity of Relative Robust Optimization}
All relative robust problems introduced above are three-level 
optimization problems. In contrast to the two-level absolute robust
models \eqref{eq:ARminvar}--\eqref{eq:ARmaxrar}, relative 
robust problems are generally intractable, even for tractable uncertainty 
sets $\cU$. This is further illustrated in Section \ref{sec:ellipsoidal} in the case 
where $\cU_{\mu}$ is an ellipsoid and $\cU_{\Cov}$ a singleton. The best 
we can hope to achieve in this case is to identify tractable approximations.  Most of 
Section \ref{sec:ellipsoidal} is therefore spent on deriving good polynomial-time 
solvable inner approximations to this problem. Outer approximations -- that is, relaxations -- 
that rely on the tractability results for polytopic uncertainty sets derived in Section \ref{sec:tope} 
were discussed by Taguchi et.\ al.\ \cite{Taguchi2} and Adom \cite{Adom}. 

\section{Finite and Polytopic Uncertainty Sets} \label{sec:tope}

In this section we will show that if the uncertainty set $\cU$ is chosen as a 
polytope -- that is, the convex hull of $k$ points -- or a set of $k$ points, then the 
relative robust optimization problems 
\eqref{relrob:minvar version1}, 
\eqref{relrob:minvar version2}, \eqref{relrob:maxret version1}, 
\eqref{relrob:maxret version2}, \eqref{eq:rrp1} and 
\eqref{eq:RRSharpe} are polynomial-time solvable as a function of $k$ and the problem 
dimension. The complexity is also polynomial in the logarithm of a condition number 
\cite{cuckerPena}, as any conic programming problem, but we will not discuss details here. 

\subsection{Solving Problem (\ref{eq:rrp1})}

We start by considering Problem \eqref{eq:rrp1} in the form \eqref{main form I} and by 
assuming that the uncertainty set 
\begin{equation*}
\cU=\{(\mu^{[i]},\Cov^{[i]}):\,(i=1,\dots,k)\}
\end{equation*} 
consists of finitely many scenarios. For any given $\mu$ and positive definite $\Cov$, $z^*(\mu, \Cov)$ is
easily computed by solving a convex quadratic optimization problem. Therefore, each instance of 
the last inequality in the formulation \eqref{main form I} is a convex quadratic 
constraint that can be efficiently handled using, for example, conic optimization methods. 
Model \eqref{main form I} can thus be rewritten by enumerating the possibilities, 
\begin{align}
\min_{x,\gamma}\;&\gamma\label{eq:rrp1finite}\\
&x\in\cX,\nonumber\\
&\mu^{[i]\T}x-\lambda x^{\T}\Cov^{[i]}x\geq z^*(\mu^{[i]},\Cov^{[j]})-\gamma,\quad
(i=1,\dots k)\nonumber
\end{align}
This shows that the relative robust problem \eqref{eq:rrp1} can be solved
by first obtaining the optimal values $z^*(\mu^{[i]},\Cov^{[i]})$ and then
solving problem \eqref{eq:rrp1finite} as a second-order cone programming problem (SOCP) 
with $k$ convex quadratic constraints. While this process may be tedious and time consuming, the 
resulting formulation is a single level deterministic optimization problem that can be solved 
efficiently for realistic problem sizes both in terms of the dimension and the number of parameter scenarios. 

Next, we consider polytopic uncertainty sets, namely those defined as the
convex hull of a finite number of extreme scenarios, 
\begin{equation}\label{given as}
\cU=\conv\bigl(\{(\mu^{[i]},\Cov^{[i]}):\,i=1,\dots,k\}\bigr). 
\end{equation}
Using Lemma \ref{lem:simple} we can now observe that the relative robust 
model \eqref{eq:rrp1} that corresponds to this uncertainty set is also solved by 
\eqref{eq:rrp1finite}. This follows immediately from the following corollary:

\begin{corollary} \label{lem:polytopic}
For $\cU$ given in \eqref{given as} and $x\in\R^n$, the following are equivalent, 
\begin{itemize}
\item[i) ] $\mu^{\T} x - \lambda x^{\T}\Cov x \geq z^*(\mu,\Cov) - \gamma$ for all 
$(\mu,\Cov)\in\cU$, 
\item[ii) ]  $\mu^{[i]\T} x - \lambda x^{\T}\Cov^{[i]} x \geq z^*(\mu^{[i]},\Cov^{[i]}) - \gamma$ 
for $(i=1,\dots,k)$. 
\end{itemize}
\end{corollary}

\begin{proof}
We only need to show that ii) implies i), as the reverse implication is trivial. 
For each $(\mu,\Cov)\in\cU$ there exist weights $\alpha^{[i]}\geq 0$ such that 
$\sum_{i=1}^k\alpha^{[i]}=1$ and $(\mu,\Cov)=\sum_{i=1}^k\alpha^{[i]}(\mu^{[i]},
\Cov^{[i]})$. Multiplying each inequality in ii) by $\alpha^{[i]}$ and taking the sum, 
one obtains the required inequality 
\begin{equation*}
\mu^{\T}x - \lambda x^{\T}\Cov x\geq 
\sum_{i=1}^k\alpha^{[i]} z^*(\mu^{[i]},\Cov^{[i]})-\gamma
\geq z^*(\mu,\Cov ) -\gamma, 
\end{equation*}
where the second inequality follows from the linearity of the function $(\mu,\Cov)\mapsto 
\mu^{\T}x-\lambda x^{\T}\Cov x$ and application of Lemma \ref{lem:simple}. 
\end{proof}

Thus, when the uncertainty set is given as a convex hull the relative robust problem
\eqref{eq:rrp1} is tractable. Similar results hold for Models \eqref{relrob:minvar version1}, 
\eqref{relrob:minvar version2}, \eqref{relrob:maxret version1} and 
\eqref{relrob:maxret version2}. The reader will find it easy to work out the details.

\subsection{Solving Problem (\ref{eq:RRSharpe})}\label{ssec:Sharpe}

Recall that in the case of Model \eqref{eq:RRSharpe}, we assumed the uncertainty set to be of 
the form $\cU=\cU_{\mu}\times\{\Cov\}$, that is, we assumed the covariance matrix $\Cov$ to 
be known with certainty. In the case where $\cU_{\mu}$ is once again given by a finite number 
of extreme scenarios
\begin{equation*}
\cU_{\mu}=\bigl\{\mu^{[i]}:\,i=1,\dots,k\bigr\},
\end{equation*}
the $k$ values $z(\mu^{[i]})$ need to be computed by solving convex quadratic programming 
problems of the form \eqref{convexification}, and then \eqref{eq:RRSharpe} turns into the convex 
quadratic programming problem 
\begin{align}
\min_{\gamma\in\R, y\in\R^n}\,&\gamma\label{aber hallo}\\
\text{s.t. }&(\mu^{[i]} -r  e)^{\T}y\geq z^*(\mu^{[i]}) -\gamma,\quad (i=1,\dots,k)\nonumber\\
&y\in\R_+\cX,\nonumber\\
&y^{\T}\Cov y\leq 1.\nonumber
\end{align}

In the case where $\cU_{\mu}$ is a convex hull of extreme scenarios 
\begin{equation*}
\cU_{\mu}=\conv\bigl(\bigl\{\mu^{[i]}:\,i=1,\dots,k\bigr\}\bigr), 
\end{equation*}
one can once again exploit the fact that the function $\mu\mapsto z^*(\mu)$ is convex, so that 
Lemma \ref{lem:simple} implies that \eqref{aber hallo} solves Problem \eqref{eq:RRSharpe}. 

\section{Ellipsoidal Uncertainty Sets}\label{sec:ellipsoidal}

In the remaining sections we assume the covariance matrix $\Cov$ to be known with certainty. 
The vector of expected returns $\mu$ is assumed to be uncertain and lie in an ellipsoidal 
uncertainty set 
\begin{equation*}
\cU_{\mu}=\left\{\overline{\mu}+Mu:\, \|u\|\leq 1\right\},
\end{equation*}
where $M$ is a $n\times k$ matrix with $k\leq n$. Although the ellipsoid $\cU$ need not be 
full-dimensional, in applications it is often natural to choose $k=n$. Throughout this section we treat 
$u$ or $\mu=\mu(u)=\overline{\mu}+Mu$ as the vector of uncertain model parameters interchangeably. Further, we assume that the set of feasible decision vectors is of the form 
\begin{equation*}
\cX=\left\{x\in\R^n:\,Fx=f,\,Gx\leq g\right\}, 
\end{equation*}
where $F\in\R^{m_f\times n}$ has full row rank, $f\in\R^{m_f}$, $G\in\R^{m_g\times n}$ and  $g\in\R^{m_g}$. We write $F_i$ and $G_i$ for the $i$-th rows of $F$ and $G$ respectively, and 
\begin{align*}
f_{\mu}(x)&:=\mu^{\T}x-\lambda x^{\T}\Cov x,\\
z^*(\mu)&:=\max_{y\in \cX}f_{\mu}(y),\\
R(x)&:=\max_{\mu\in\cU_{\mu}}z^*(\mu)-f_{\mu}(x).
\end{align*} 
The relative robust problem we wish to solve is then given by 
\begin{equation*}
\text{(RRP)}\quad\min_{x\in\cX}\,R(x).
\end{equation*}

\subsection{Copositivity Cones}\label{sec:copositivity}

We begin by introducing the technical tools that make an analysis 
of (RRP) possible. Most of the notation is adopted from the elucidating 
paper of Sturm and Zhang \cite{sturm}. 

If $D$ is subset of $\R^n$, let 
\begin{equation*}
{\mathscr H}(D):=\closure\left\{z=\left[\begin{smallmatrix}x\\ \tau
\end{smallmatrix}\right]\in\R^{n+1}:\,\tau>0,\,\tau^{-1}x\in D\right\}
\end{equation*}
be its homogenization, where $\closure(\cdot)$ denotes the topological 
closure. Let $q:x\mapsto x^{\T}Ax+2b^{\T}x+c$ 
be an arbitrary quadratic polynomial on $\R^n$, where 
$A\in\Sym{n}$ is a symmetric $n\times n$ matrix, 
$b\in\R^n$ and $c\in\R$, and let 
\begin{equation*}
{\mathscr M}(q):=\begin{bmatrix}A&b\\b^{\T}&c\end{bmatrix}.
\end{equation*}
Then 
\begin{equation*}
q(x)=\left[\begin{smallmatrix}x\\1\end{smallmatrix}\right]^{\T}
{\mathscr M}(q)\left[\begin{smallmatrix}x\\1\end{smallmatrix}\right]
=\left[\begin{smallmatrix}x\\1\end{smallmatrix}\right]
\left[\begin{smallmatrix}x\\1\end{smallmatrix}\right]^{\T}\bullet
{\mathscr M}(q), 
\end{equation*}
where we write $X\bullet Y=\trace(X^{\T}Y)$ for the 
trace inner product of two matrices of equal size (this inner product 
is the polarization 
of the Frobenius norm). In what follows we will refer to ${\mathscr M}(q)$ 
as the {\em matrix representation} of $q(\cdot)$. This defines a 
1-1 correspondence between the set of quadratic functions on $\R^n$ and 
the set $\Sym{n+1}$ of symmetric matrices of size $(n+1)$. In the 
sequel we will write $X\succeq 0$ if $X$ is a symmetric 
positive semidefinite matrix and 
\begin{equation*}
{\mathscr S}_+^{(n+1)}:=\{X\in\Sym{n+1}:\,X\succeq 0\}. 
\end{equation*}
Let 
\begin{align*}
\mathscr{FC}_+(D)&:=\left\{{\mathscr A}\in\Sym{n+1}:\,
[\begin{smallmatrix}x\\1\end{smallmatrix}]^{\T}
{\mathscr A}[\begin{smallmatrix}x\\1\end{smallmatrix}]\geq 0\;\;
\forall\,x\in D\right\}\\
&=\left\{{\mathscr A}\in\Sym{n+1}:
\,z^{\T}{\mathscr A}z\geq 0\;\;\forall\,z\in{\mathscr H}(D)\right\}
\end{align*}
be the set of quadratic functions that are nonnegative on $D$. In the 
case where $D=\{x:\,q(x)\geq 0\}$ for some quadratic polynomial 
$q$, $\mathscr{FC}_+(D)$ is called the set of quadratic functions 
copositive with $q$. Here we use an abuse of language and 
speak of $\mathscr{FC}_+(D)$ as the {\em copositivity cone} associated 
with $D$ when $D$ is a more general set.\\

\begin{lemma}[Corollary 1, \cite{sturm}]\label{lem:sturm}
$\mathscr{FC}_+(D)=\conv\{zz^{\T}:\,z\in{\mathscr H}(D)\}^*$.\\
\end{lemma}

\begin{proof}
Lemma \ref{lem:sturm} is the same as Corollary 1 in 
Sturm \& Zhang \cite{sturm}. 
Here we give an alternative proof for completeness. We have 
\begin{align*}
\mathscr{FC}_+(D)&=\{X\in\Sym{n+1}:\,
z^{\T}Xz\geq0,\;\forall\,z\in{\mathscr H}(D)\}\\
&=\bigcap_{z\in{\mathscr H}(D)}\{X\in\Sym{n+1}:\,
\langle X;zz^{\T}\rangle\geq 0\}\\
&=\left(\conv\{zz^{\T}:\,z\in{\mathscr H}(D)\}\right)^*.
\end{align*}
\end{proof}
\hspace{1cm}\\

Unfortunately, for general $D$, the cone $\conv\{zz^{\T}:\,z\in{\mathscr H}(D)\}^*$ 
does not have a tractable characterization. For example, when 
$D=\R^n_+$ then 
\begin{equation*}
{\mathscr H}(D)=\closure\left(\left\{\left[\begin{smallmatrix}x\\ \tau
\end{smallmatrix}\right]:\,\tau>0, \tau^{-1}x\in\R^n_+\right\}\right)
=\R^{n+1}_+, 
\end{equation*}
and 
\begin{equation*}
\mathscr{FC}_+(D)=\conv\{zz^{\T}:\,z\in\R^{n+1}_+\}^*
\end{equation*}
is the {\em co-positive cone}. Testing whether a given matrix belongs 
to this cone is co-NP-hard \cite{murty-kabadi}. We will see in the sequel that (RRP) is equivalent 
to solving a conic optimization problem with a conic constraint of type 
$\mathscr{FC}_+(D)$ for a convex set $D$ defined by multiple linear 
and one quadratic inequality. 
If the cone $\mathscr{FC}_+(D)$ is intractable, then the 
conic formulation of (RRP) is intractable too. To render the relative robust 
approach computationally viable in this situation, we will identify a 
tractable convex cone 
\begin{equation*}
K\subseteq\mathscr{FC}_+(D)
\end{equation*}
which can be used in an inner approximation of (RRP). 
Most of the technical details regarding the construction of $K$ are 
deferred to Section \ref{appendix}

\subsection{A Conic Formulation of (RRP)}\label{sec:reformulation}

Introducing an artificial variable $\gamma$, (RRP) is easily 
seen to be equivalent to 
\begin{align*}
\text{(RRP.i)}\quad\min_{x,\gamma}\,&\gamma\\
\text{s.t.}\quad&x\in\cX\\
&\gamma\geq z^*(\mu)-f_{\mu}(x)\quad\forall\mu\in
\cU_{\mu}.
\end{align*}
Since $z^*(\mu):=\max_{y\in \cX}f_{\mu}(y)$, this problem can be further 
rewritten as 
\begin{align}
\text{(RRP.ii)}\qquad\min_{x,\gamma}\,&\gamma\nonumber\\
\text{s.t.}\quad&x\in\cX\nonumber\\
&\gamma\geq f_{\mu}(y)-f_{\mu}(x)\quad\forall\mu\in
\cU_{\mu},\;y\in\cX.\label{semi-infinite}
\end{align}
This is a semi-infinite optimization problem, that is, a finite-dimensional 
problem with infinitely many constraints. To render this problem amenable 
to numerical computations, we have to replace these infinitely many 
constraints by a finitely many. 

Parameterizing $\mu$ by $u$, the set of values of $(\mu(u),y)$ that 
appear in the right-hand side of \eqref{semi-infinite} corresponds to 
\begin{equation}\label{set of values}
D:=\left\{\left[\begin{smallmatrix}u\\y\end{smallmatrix}\right]:\,
u^{\T}u\leq 1,\,Fy=f,\,Gy\leq g\right\}. 
\end{equation}
It follows from Lemma 4 of Sturm-Zhang \cite{sturm} that the 
homogenization of this set is characterized by 
\begin{align*}
{\mathscr H}(D)=\left\{\left[\begin{smallmatrix}u\\y\\ \tau\end{smallmatrix}
\right]\right.&\in\R^{k+n+1}:\,\tau\geq 0,\, 
\left[\begin{smallmatrix}u\\y\\ \tau\end{smallmatrix}\right]^{\T}
\left[\begin{smallmatrix}-\I&0&0\\0&0&0\\0&0&1\end{smallmatrix}
\right]\left[\begin{smallmatrix}u\\y\\ \tau\end{smallmatrix}\right]\geq 0,\\
&\left[\begin{smallmatrix}0\\-F_i\\f_i\end{smallmatrix}\right]^{\T}
\left[\begin{smallmatrix}u\\y\\ \tau\end{smallmatrix}\right]=0,\,(i=1,\dots,m_f),\\
&\left.\left[\begin{smallmatrix}0\\-G_i\\g_i\end{smallmatrix}\right]^{\T}
\left[\begin{smallmatrix}u\\y\\ \tau\end{smallmatrix}\right]\geq 0,\,(i=1,\dots,m_g)
\right\}.
\end{align*}
Further, for fixed $(x,\gamma)$, the expression 
\begin{align*}
q_{x,\gamma}(u,y)&:=\gamma-f_{\mu}(y)+f_{\mu}(x)\\
&=\left[\begin{smallmatrix}u\\y\end{smallmatrix}\right]^{\T}
\left[\begin{smallmatrix}0&-\Half M^{\T}\\
-\Half M&\lambda \Cov\end{smallmatrix}\right]
\left[\begin{smallmatrix}u\\y\end{smallmatrix}\right]
+\left[\begin{smallmatrix}M^{\T}x\\- \overline{\mu}\end{smallmatrix}\right]^{\T}
\left[\begin{smallmatrix}u\\y\end{smallmatrix}\right]
+(\gamma-\lambda x^{\T}\Cov x+ \overline{\mu}^{\T}x)
\end{align*}
is a quadratic function of $(u,y)$ whose matrix representation in the 
homogenized space is given by 
\begin{equation*}
{\mathscr M}_{x,\gamma}:=
{\mathscr M}(q_{x,\gamma})=\left[\begin{smallmatrix}0&-\Half M^{\T}&\Half M^{\T}x\\
-\Half M&\lambda \Cov &-\Half\overline{\mu}\\
\Half x^{\T}M&-\Half\overline{\mu}^{\T}&(\gamma-\lambda x^{\T}\Cov x+
\overline{\mu}^{\T}x)\end{smallmatrix}\right].
\end{equation*}
Using Lemma \ref{lem:sturm}, Condition \eqref{semi-infinite} is seen to 
be the same as 
\begin{equation}\label{curly M}
{\mathscr M}_{x,\gamma}\in\conv\left\{zz^{\T}:\,z\in{\mathscr H}(D)\right\}^*.
\end{equation}
Therefore, (RRP.ii) can be written in conic form,
\begin{align*}
\text{(RRP.iii)}\qquad\min_{x,\gamma}\,&\gamma\\
\text{s.t.}\quad&x\in\cX,\\
&{\mathscr M}_{x,\gamma}\in\conv\left\{zz^{\T}:\,z\in{\mathscr H}(D)\right\}^*.
\end{align*}

\subsection{A Tractable Inner Approximation}\label{tractable inner}

Let $H$ be the trailing $n\times(n-m_f)$ block of 
the orthogonal factor in the QR-decomposition $[\begin{smallmatrix}\star&H
\end{smallmatrix}]R$ of $F^{\T}$, so that the columns of $H$ form a 
basis of $\ker(F)$. Further, let $x_p\in\R^n$ be a particular solution of 
the system $Fy=f$, and let us write $r:=n-m_f$, so that 
\begin{equation*}
\{x\in\R^n:\,Fx=f\}=\{x_p+H w:\,w\in\R^{r}\}. 
\end{equation*}
Let $p_0^{\T}=[\begin{smallmatrix}0&1\end{smallmatrix}]$, where 
$0$ is a zero row vector of size $k+r$, and let $p_i^{\T}$ $(i=1,\dots,m_g)$ 
be the row vectors of the matrix $[\begin{smallmatrix}0&-GH&g-Gx_p
\end{smallmatrix}]$, where $0$ is now a zero matrix of size $m_g\times k$. 
And finally, let $\Cov=U^{\T}U$ be the Cholesky factorization of $\Cov$. With 
this notation, Corollary \ref{approximation corollary} of Section \ref{appendix} shows 
that the following problem is an inner approximation of (RRP), 
where the minimisation is over the decision variables 
$w\in\R^r$, $\gamma,s,\eta,\xi_{ij}\in\R$, $(i\neq j=0,\dots,m_g)$, and 
$\tau_i\in\R$, $u_i\in\R^k$, $(i=0,\dots,m_g)$:
\begin{align*}
\text{(ARRP)}\;\min_{w,\gamma,s,\eta,\xi,\tau,u}\,&\gamma\\
\text{s.t.}\qquad&\hspace{-0.5cm}
g-Gx_p-GHw\in\R^{m_g}_+,\\
&\hspace{-0.5cm}\eta,\xi_{ij}\in\R_+,\quad (i\neq j=0,\dots,m_g),\\
&\hspace{-0.5cm}[\begin{smallmatrix}\tau_i\\u_i\end{smallmatrix}]\in 
L_{k+1},\quad (i=0,\dots,m_g),\\
&\hspace{-.5cm}\begin{bmatrix}0\\0\\Ux_p\end{bmatrix}
+\begin{bmatrix}\frac{1}{\sqrt{2}}&\frac{1}{\sqrt{2}}&\\
\frac{1}{\sqrt{2}}&-\frac{1}{\sqrt{2}}\\&&UH\end{bmatrix}
\begin{bmatrix}\Half\\s\\w\end{bmatrix}\in L_{n+2},\\
&\hspace{-0.5cm}\begin{bmatrix}0&-\Half M^{\T}H&\Half M^{\T}Hw\\
-\Half H^{\T}M&\lambda H^{\T}\Cov H&H^{\T}(\lambda \Cov x_p
-\Half\overline{\mu})\\
\Half w^{\T}H^{\T}M&(\lambda \Cov x_p-\Half\overline{\mu})^{\T}H&
\gamma-\lambda s+(\overline{\mu}-2\lambda \Cov x_p)^{\T}H
w\end{bmatrix}\\
&\hspace{1.5cm}-\eta\left[\begin{smallmatrix}-\I&&\\&0_r&\\&&1\end{smallmatrix}\right]
-\sum_{i\neq j=0}^{m_g}\xi_{ij}\left(p_i p_j^{\T}+p_j p_i^{\T}\right)\\
&\hspace{3.5cm}+\sum_{i=0}^{m_g}\left(p_i\left[\begin{smallmatrix}u_i\\0\\ \tau_i
\end{smallmatrix}\right]^{\T}+\left[\begin{smallmatrix}u_i\\0\\ \tau_i
\end{smallmatrix}\right]p_i^{\T}\right)\in\Sym{k+r+1}_+.
\end{align*}
By construction, every solution $(w,\gamma,s,\eta,\xi,\tau,u)$ to (ARRP) 
provides a feasible solution $(x_p+Hw,\gamma)$ to (RRP.iii). Since 
the feasible set of (ARRP) is thus smaller than the feasible set of (RRP.iii), 
a (ARRP)-optimal solution $(w^*,\gamma^*,s^*,\eta^*,\xi^*,\tau^*,u^*)$ 
does not necessarily correspond to a (RRP)-optimal 
is optimal for (ARRP), this does 
not necessarily imply that $(x_p+Hw^*,\gamma^*)$ is optimal for (RRP.iii). 
However, since (ARRP) is equivalent to (RRP.iii) in the case $m_g\in\{0,1\}$ 
(see Corollary \ref{approximation corollary}), 
it is reasonable to expect that $(x_p+Hw^*,\gamma^*)$ is a feasible 
solution to (RRP.iii) which is quite close to optimal even in the case 
$m_g\geq 2$. 

Note that each of the constraints of (ARRP) is formulated as a conic inequality of an expression 
that is linear in the decision variables. 
Thus, the great advantage of working with the model (ARRP) rather than (RRP) is the fact 
that, while (RRP) may be NP-hard, (ARRP) is readily solvable via standard 
polynomial-time conic programming implementations such as SDTP3 
\cite{sdtp3} or Sedumi \cite{sedumi}. 

\section{Tightness of Inner Approximations}\label{appendix}

In this section we discuss some of the technical details and tightness results surrounding 
the inner approximation of the cones $\mathscr{FC}_+(D)$ used in Section \ref{sec:ellipsoidal}. 

\subsection{The Case of General $D$}\label{arbitrary}

We begin with the discussion of inner approximations of $\mathscr{FC}_+(D)$ 
where $D\subseteq\R^n$ is an arbitrary set.

\begin{lemma}\label{lem:inner}
$\mathscr{FC}_+(D)^*\subseteq\{X\in\Sym{n+1}_+:\,
Xw\in{\mathscr H}(D)\;\forall\,w\in{\mathscr H}(D)^*\}$.
\end{lemma}

\begin{proof}
By Lemma \ref{lem:sturm}, any $X\in\mathscr{FC}_+(D)^*$ can be 
written as a limit $X=\lim_{j\rightarrow\infty}X_j$, where 
$X_j=\sum_{i=1}^{k_j}\xi_{ij} z_{ij} z_{ij}^{\T}$ for some 
$z_{ij}\in{\mathscr H}(D)$ and $\xi_{ij}\geq 0$ $(i=1,\dots,k_j)$. 
Clearly this implies that $X\succeq 0$, and since for any 
$w\in{\mathscr H}(D)^*$ we have 
\begin{equation*}
X_jw=\sum_{i=1}^{k_j}\xi_{ij}(z_{ij}^{\T}w)z_{ij}\in{\mathscr H}(D)
\end{equation*}
and ${\mathscr H}(D)$ is closed, $Xw\in{\mathscr H}(D)$. 
\end{proof}

Taking duals in the inclusion of Lemma \ref{lem:inner}, we obtain the 
following inner approximation of $\mathscr{FC}_+(D)$,
\begin{equation}\label{round 1}
{\mathscr FD}_+(D)\supseteq\Sym{n+1}+K^*,
\end{equation}
where 
\begin{equation}\label{K}
K=\{X\in\Sym{n+1}:\,Xw\in{\mathscr H(D)}\;
\forall\,w\in{\mathscr H}(D)^*\}.
\end{equation}
To make this result useful, we need to characterize $K^*$.

\begin{lemma}\label{kw**}
Let $C\subseteq\R^n$ be a closed convex cone and $w\in\R^n$. 
Then 
\begin{equation*}
\left\{X\in\Sym{n}:\,Xw\in C\right\}^*=\left\{vw^{\T}+wv^{\T}:\,
v\in C^*\right\}. 
\end{equation*}
\end{lemma}

\begin{proof}
Consider the linear map $\varphi_w(X)=Xw$ from $\Sym{n}$ to 
$\R^{n}$. Endowing these spaces with their canonical inner products 
$\langle X,Y\rangle:=X\bullet Y$ and $\langle x,y\rangle:=x^{\T}y$, 
the adjoint map $\varphi_w^*:\R^{n+1}\rightarrow\Sym{n+1}$ is 
defined by the relation 
\begin{equation*}
\langle X,\varphi_w^*(v)\rangle=\langle\varphi_w(X),v\rangle,\quad
(X\in\Sym{n+1}, v\in\R^{n+1}).
\end{equation*}
The right-hand side in this equation equals 
$(Xw)^{\T}v=\Half\langle X,(wv^{\T}+vw^{\T})\rangle$, showing that 
\begin{equation}\label{adjoint}
\varphi_w^*(v)=\frac{1}{2}(wv^{\T}+vw^{\T}).
\end{equation} 
Now we have $Xw\in C$ if and only if 
$\langle\varphi_w(X),v\rangle\geq 0$ for all $v\in C^*$ (using biduality and 
the assumption that $C$ is a closed convex cone). Taking adjoints, 
this is further equivalent to 
$\langle X,\varphi^*_w(v)\rangle\geq 0$ for all $v\in C^*$, and finally to 
\begin{equation}\label{before dual}
X\in(\varphi_w^*(C^*))^*. 
\end{equation}
Since $C^*$ is a closed convex cone and 
$\varphi_w^*$ is a linear map between finite-dimensional vector spaces, 
$\varphi_w^*(C^*)$ is a closed convex cone, so that taking duals in 
\eqref{before dual} yields 
\begin{equation*}
\left\{X\in\Sym{n}:\,Xw\in C\right\}^*=\varphi_w^*(C^*).
\end{equation*}
Using \eqref{adjoint}, this is seen to be equivalent to the claim 
of the lemma. 
\end{proof}

\begin{lemma}\label{lem:k*}
Let $K$ be the cone defined in \eqref{K}. Then 
\begin{equation*}
K^*=\closure\left(\cone\{wv^{\T}+vw^{\T}:\,v,w\in{\mathscr H}(D)^*\}
\right).
\end{equation*}
\end{lemma}

\begin{proof}
For each $w\in{\mathscr H}(D)^*$, let 
$K_w:=\{X\in\Sym{n}:\,Xw\in{\mathscr H}(D)\}$. Since 
${\mathscr H}(D)$ is a closed convex cone, Lemma \ref{kw**} 
shows that $K_w^*=\{vw^{\T}+wv^{\T}:\,v\in{\mathscr H}(D)^*\}$. 
Therefore, we have 
\begin{align*}
K^*=&\bigl(\bigcap_{w\in{\mathscr H}(D)^*}K_w\bigr)^*\\
&=\closure\Bigl(\cone\bigl(\bigcup_{w\in{\mathscr H}(D)^*}K_w^*
\bigr)\Bigr)\\
&=\closure\bigl(\cone\{wv^{\T}+vw^{\T}:\,v,w\in{\mathscr H}(D)^*\}
\bigr),
\end{align*}
as claimed. 
\end{proof}

\begin{example}\label{example 1}
Let $D=\{x\in\R^n:\,b^{\T}x\geq c\}$. Then ${\mathscr H}(D)=
\{z\in\R^{(n+1)}:\,a^{\T}z\geq 0\}$, where 
$a=[\begin{smallmatrix}b\\-c\end{smallmatrix}]$ and 
${\mathscr H}(D)^*=\cone\{a\}$. Furthermore, we have 
\begin{align*}
K&=\{X\in\Sym{n+1}:\,Xa\in{\mathscr H}
(D)\}\\
&=\{X\in\Sym{n+1}:\, a^{\T}Xa\geq 0\}\\
&=\{X\in\Sym{n+1}:\,X\bullet aa^{\T}\geq 0\},
\end{align*}
so that $K^*=\cone\{aa^{\T}\}$. This is confirmed by Lemma \ref{lem:k*}, 
which says that $K^*=\cone\{wv^{\T}+vw^{\T}:\,v,w\in\cone\{a\}\}=
\cone\{aa^{\T}\}$.
\end{example}

\begin{example}\label{example 2}
Let $D=\{x\in\R^n:\,\|x\|\leq 1\}$. Then 
\begin{equation*}
{\mathscr H}(D)=\left\{z\in\R^{n+1}:\,z^{\T}\left[\begin{smallmatrix}
-\I&0\\0&1\end{smallmatrix}\right]z\geq 0,\,\left[\begin{smallmatrix}
0\\1\end{smallmatrix}\right]^{\T}z\geq 0\right\}=\left[\begin{smallmatrix}
0&\I\\ 1&0\end{smallmatrix}\right]L_{n+1},
\end{equation*}
where $L_{n+1}$ is the $n+1$-dimensional Lorenz cone or second-order 
cone and the operator $\bigl[\begin{smallmatrix}0&1\\ \I&0\end{smallmatrix}
\bigr]$ permutes the first component of a vector into last place. 
Since $L_{n+1}$ is self-dual, we have 
\begin{equation}\label{6.**}
K=\{X\in\Sym{n+1}:\,Xw\in \left[\begin{smallmatrix}
0&\I\\ 1&0\end{smallmatrix}\right]L_{n+1}\;\forall\,w\in
\left[\begin{smallmatrix}
0&\I\\ 1&0\end{smallmatrix}\right]L_{n+1}\}.
\end{equation}
Lemma \ref{lem:k*} thus shows that $K^*=\cone\{wv^{\T}+vw:\,v,w
\in \left[\begin{smallmatrix}
0&\I\\ 1&0\end{smallmatrix}\right]L_{n+1}\}$.
\end{example}

Combining the inclusion \eqref{round 1} with Lemma \ref{lem:k*}, we 
arrive at the following result.

\begin{theorem}\label{prop:copositive}
For any set $D\subseteq\R^n$ it is true that 
\begin{equation*}
\mathscr{FC}_+(D)\supseteq\Sym{n+1}+
\closure\bigl(\cone\{wv^{\T}+vw^{\T}:\,v,w\in{\mathscr H}(D)^*\}\bigr).
\end{equation*}
\end{theorem}

\subsection{Approximation Tightness in the General Case}\label{tightness}

The inner approximation of Theorem \ref{prop:copositive} is valid for 
arbitrary $D\subset\R^n$, and when ${\mathscr H}(D)$ can be explicity 
characterized it becomes a useful computational tool in conjunction with 
Carath\'eodory's theorem. It is therefore natural to ask if the inclusion 
given by the theorem is in fact an equality. Unfortunately, this is not 
true in general, as we shall now see. 

\begin{example}\label{example3}
Consider Example \ref{example 2} 
again, and note that in this case any $z\in{\mathscr H}(D)=L_{n+1}$ satisfies 
\begin{equation*}
\left[\begin{smallmatrix}-\I&0\\0&1\end{smallmatrix}\right]\bullet zz^{\T}
=z^{\T}\left[\begin{smallmatrix}-\I&0\\0&1\end{smallmatrix}\right]z
\geq 0,
\end{equation*}
and the same holds true for convex combinations of matrices $zz^{\T}$ with 
$z\in\left[\begin{smallmatrix}
0&\I\\ 1&0\end{smallmatrix}\right]L_{n+1}$. 
Thus, using this extra information, we could have used the 
tighter approximation 
\begin{equation*}
\mathscr{FC}_+(D)^*\subseteq\Sym{n+1}_+\cap K \cap C,
\end{equation*}
where $K$ is defined as in \eqref{6.**} and 
\begin{equation*}
C:=\left\{X\in\Sym{n+1}:\,
\left[\begin{smallmatrix}-\I&0\\0&1\end{smallmatrix}\right]\bullet 
X\geq 0\right\}.
\end{equation*}
This yields the inner approximation 
\begin{align}
\mathscr{FC}_+(D)&\supseteq\Sym{n+1}_+
+K^*+C^*\nonumber\\
&=\Sym{n+1}_+
+\cone\{wv^{\T}+vw^{\T}:\,v,w\in \left[\begin{smallmatrix}
0&\I\\ 1&0\end{smallmatrix}\right]L_{n+1}\}
+\cone\left\{\left[\begin{smallmatrix}-\I&0\\0&1
\end{smallmatrix}\right]\right\}
\label{better approx}
\end{align}
which is strictly larger than the inner approximation 
\begin{equation}\label{rhs}
\mathscr{FC}_+(D)\supseteq\Sym{n+1}_+
+\cone\{wv^{\T}+vw^{\T}:\,v,w\in \left[\begin{smallmatrix}
0&\I\\ 1&0\end{smallmatrix}\right]L_{n+1}\}
\end{equation}
provided by Theorem \ref{prop:copositive}, as the right-hand side of 
\eqref{rhs} does not contain the matrix 
$[\begin{smallmatrix}-\I&0\\0&1\end{smallmatrix}]$. 
Further, applying the s-Lemma (see Lemma \ref{s-Lemma} below) in 
the context of this Example, one finds 
\begin{equation}\label{for later use too}
{\mathscr FD}_+(D)=\Sym{n+1}_+ +
\cone\left\{\left[\begin{smallmatrix}-\I&0\\0&1\end{smallmatrix}
\right]\right\},
\end{equation}
an identity which was first discovered by Rendl-Wolkowicz \cite{rendl}. 
In other words, the approximation \eqref{better approx} is not only 
an improvement over that of Theorem \ref{prop:copositive}, but it is 
in fact tight, that is, the inclusion becomes an equality. Note that this 
also shows that 
\begin{equation}\label{for later use}
\{wv^{\T}+vw^{\T}:\,v,w\in \left[\begin{smallmatrix}
0&\I\\ 1&0\end{smallmatrix}\right]L_{n+1}\}\subset
\Sym{n+1}_+ 
+\cone\left\{\left[\begin{smallmatrix}-\I&0\\0&1\end{smallmatrix}
\right]\right\}.
\end{equation}
\end{example}

The following classical result from the theory of robust control theory 
was used in the analysis of the above example:

\begin{lemma}[s-Lemma, Yakubovich \cite{yakubovich}]
\label{s-Lemma}
If $D=\{x\in\R^n:\,q(x)\geq 0\}$, where $q(\cdot)$ is a quadratic 
function that takes a strictly positive value somewhere, then 
\begin{equation*}
{\mathscr FD}_+(D)=\Sym{n+1}_+ +
\cone\{{\mathscr M}(q)\}.
\end{equation*}
\end{lemma}

For proofs see e.g.\ \cite{yakubovich}, \cite{polik} and \cite{hauserS-Lemma}.

\subsection{Improved Approximation for Use in Section 
\ref{sec:ellipsoidal}}\label{sec:inner approx quadratic}

Next we will generalize 
the improved inner approximation \eqref{better approx} to copositivity 
cones associated with the convex set $D$ we found in \eqref{set of values}, 
\eqref{semi-infinite}, 
\begin{equation*}
D:=\left\{\left[\begin{smallmatrix}u\\y\end{smallmatrix}\right]\in\R^{k+n}:\,
u^{\T}u\leq 1,\,Fy=f,\,Gy\leq g\right\}.
\end{equation*}
Recall that $F\in\R^{m_f\times n}$ has full row rank, and 
$G\in\R^{m_g\times n}$. Let $H$ be the trailing $n\times(n-m_f)$ block of 
the Q-factor of the QR-decomposition $[\begin{smallmatrix}\star&H
\end{smallmatrix}]R$ of $F^{\T}$, so that the columns of $H$ form a 
basis of $\ker(F)$. Further, let $x_p\in\R^n$ be a particular solution of the 
system $Fy=f$, and let us write $r:=n-m_f$, so that 
\begin{equation}\label{eq:Lambda}
\{y\in\R^n:\,Fy=f\}=\{x_p+H w:\,w\in\R^{r}\}. 
\end{equation}
Consider the linear map 
\begin{align*}
\Lambda:\Sym{k+n+1}&\rightarrow{\mathscr S}^{k+r+1},\\
\begin{bmatrix}A_{11}&A_{12}&b_1\\A_{12}^{\T}&A_{22}&b_2\\
b_{1}^{\T}&b_{2}^{\T}&c\end{bmatrix}&\mapsto\begin{bmatrix}
A_{11}&A_{12}H&b_1+A_{12}x_p\\
H^{\T}A_{12}^{\T}&H^{\T}A_{22}H&H^{\T}(b_2+A_{22}x_p)\\
b_{1}^{\T}+x_p^{\T}A_{12}^{\T}&(b_2+A_{22}x_p)^{\T}H&
x_p^{\T}A_{22}x_p+2b_2^{\T}x_p+c
\end{bmatrix}.
\end{align*}
For any quadratic function $h$ on $\R^{k+n}$ with matrix 
representation ${\mathscr A}\in{\mathscr S}^{k+n+1}$ let 
$h_{\Lambda}$ be the corresponding quadratic function on 
$\R^{k+r}$ with matrix representation $\Lambda({\mathscr A})$. 
If $y=x_p+H w$, then by construction of $\Lambda$ we have 
\begin{equation}\label{if and only if}
h(u,y)\geq 0\Leftrightarrow h_{\Lambda}(u,w)\geq 0. 
\end{equation}
Let $q(u,y)=1-u^{\T}u$, so that 
\begin{equation*}
D=\left\{\left[\begin{smallmatrix}u\\y\end{smallmatrix}\right]:\,
q(u,y)\geq 0,\, \left[\begin{smallmatrix}0&F\end{smallmatrix}\right]
\left[\begin{smallmatrix}u\\y\end{smallmatrix}\right]=f,\, 
\left[\begin{smallmatrix}0&G\end{smallmatrix}\right]
\left[\begin{smallmatrix}u\\y\end{smallmatrix}\right]\leq g\right\},
\end{equation*}
and note that 
\begin{equation*}
{\mathscr M}(q)=\left[\begin{smallmatrix}-\I&&\\&0_n&\\&&1\end{smallmatrix}
\right],\quad
{\mathscr M}(q_{\Lambda})=\left[\begin{smallmatrix}-\I&&\\&0_r&\\&&1
\end{smallmatrix}\right],
\end{equation*}
where $\I$ is an identity matrix of size $k$ and $0_n$, $0_r$ are zero 
matrices of size $n$ and $r$ respectively. Let 
\begin{equation*}
D_{\Lambda}:=\left\{\left[\begin{smallmatrix}u\\w\end{smallmatrix}\right]
\in\R^{k+r}:\,q_{\Lambda}(u,w)\geq 0,\, 
\left[\begin{smallmatrix}0&GH\end{smallmatrix}\right]
\left[\begin{smallmatrix}u\\w\end{smallmatrix}\right]\leq g-Gx_p\right\}.
\end{equation*}
We now obtain the following result, which shows that we can work directly 
in the reduced space $\R^{k+r}$ instead of $\R^{k+n}$: 

\begin{proposition}\label{prop:subspace}
\begin{itemize}
\item[i)\;] $D=\{[\begin{smallmatrix}u&(x_p+Hw)^{\T}\end{smallmatrix}
]^{\T}:\,[\begin{smallmatrix}u&w\end{smallmatrix}]^{\T}\in D_{\Lambda}\}$, 
\item[ii)\;] $\mathscr{FC}_+(D)=\Lambda^{-1}(\mathscr{FC}_+(D_{\Lambda}))$.
\end{itemize}
\end{proposition}

\begin{proof}
i) follows from \eqref{eq:Lambda} and \eqref{if and only if}, while ii) follows 
from part i) and \eqref{if and only if}.
\end{proof}

Next, let $p_0^{\T}=[\begin{smallmatrix}0&1\end{smallmatrix}]$, where 
$0$ is a zero row vector of size $k+r$, and let $p_i^{\T}$ $(i=1,\dots,m_g)$ 
be the row vectors of the matrix $[\begin{smallmatrix}0&-GH&g-Gx_p
\end{smallmatrix}]$, where $0$ is now a zero matrix of size $m_g\times k$. 
Then it follows from Lemma 4 in Sturm-Zhang \cite{sturm} that 
\begin{equation}\label{st shows}
{\mathscr H}(D_{\Lambda})=\left\{z\in\R^{k+r+1}:\,
z^{\T}\left[\begin{smallmatrix}-\I&&\\&0_r&\\&&1
\end{smallmatrix}\right]z\geq 0,\, p_i^{\T}z\geq 0,\, (i=0,\dots,m_g)\right\}.
\end{equation}

\begin{theorem}\label{improved approximation theorem}
An inner approximation of the cone $\mathscr{FC}_+(D_{\Lambda})$ is 
given by 
\begin{multline}\label{eq:general right hand side}
{\mathscr FC}_+(D_{\Lambda})\supseteq
\Sym{k+r+1}_+ 
+\cone\left\{\left[\begin{smallmatrix}-\I&&\\&0_r&\\&&1
\end{smallmatrix}\right]\right\}\\
+\cone\left\{p_ip_j^{\T}+p_jp_i^{\T}:\,i\neq j\in\{0,\dots,m_g\}\right\}\\
+\sum_{i=0}^{m_g}\left\{p_i
\left[\begin{smallmatrix}u\\0\\ \tau\end{smallmatrix}\right]^{\T}
+\left[\begin{smallmatrix}u\\0\\ \tau\end{smallmatrix}\right]p_i^{\T}:\,
\left[\begin{smallmatrix}\tau\\u\end{smallmatrix}\right]\in 
L_{k+1}\right\}.
\end{multline}
Furthermore, if $m_g\in\{0,1\}$ then the inclusion 
in \eqref{eq:general right hand side} is an equality. 
\end{theorem}

\begin{proof}
Lemma \ref{lem:sturm} and \eqref{st shows} establish 
\begin{align}
{\mathscr FC}_+(D_{\Lambda})^*&=\conv\left\{zz^{\T}:\,z\in\R^{k+r+1},\,
z^{\T}\left[\begin{smallmatrix}-\I&&\\&0_r&\\&&1
\end{smallmatrix}\right]z\geq 0,\, p_i^{\T}z\geq 0,\, (i=0,\dots,m_g)\right\}
\nonumber\\
&\subseteq\Sym{k+r+1}_+ \cap K^q \cap K^l_0\cap\dots\cap K^l_{m_g}
\cap K_{\mathscr H},\label{cone intersect}
\end{align}
where 
\begin{align*}
K^{q}&:=\left\{X\in\Sym{k+r+1}:\,X\bullet \left[\begin{smallmatrix}
-\I&&\\&0_r&\\&&1\end{smallmatrix}\right]\geq 0\right\},\\
K^{l}_i&:=\left\{X\in\Sym{k+r+1}:\,X p_i\in{\mathscr H}(D_{\Lambda})
\right\},\quad(i=0,\dots,m_g),\\
 K_{\mathscr H}&:=\left\{X\in\Sym{k+r+1}:\,Xv\in{\mathscr H}(D_{\Lambda})
 \;\forall\,v\in{\mathscr H}(D_{\Lambda})^*\right\}.
\end{align*}
Using the self-duality of $L_{k+1}$ and \eqref{st shows}, we get 
\begin{equation*}
{\mathscr H}(D_{\Lambda})^*=\cone\left\{p_i:\,i=0,\dots,m_g\right\}
+\left\{\left[\begin{smallmatrix}u\\0\\ \tau\end{smallmatrix}\right]:\,
\left[\begin{smallmatrix}\tau\\ u\end{smallmatrix}\right]\in L_{k+1}
\right\},
\end{equation*}
so that Lemma \ref{lem:k*} implies 
\begin{multline}\label{obige}
K_{{\mathscr H}}^*=\cone\left\{p_i p_j^{\T}+p_j p_i^{\T}:\,i,j\in\{0,\dots,
m_g\}\right\}\\
+\left\{\left[\begin{smallmatrix}u\\0\\ 
\tau\end{smallmatrix}\right]\left[\begin{smallmatrix}v\\0\\ 
\sigma\end{smallmatrix}\right]^{\T}+\left[\begin{smallmatrix}v\\0\\ 
\sigma\end{smallmatrix}\right]\left[\begin{smallmatrix}u\\0\\ 
\tau\end{smallmatrix}\right]^{\T}:\,\left[\begin{smallmatrix}\tau\\u
\end{smallmatrix}\right], \left[\begin{smallmatrix}\sigma\\v
\end{smallmatrix}\right]\in L_{k+1}\right\}\\
+\sum_{i=0}^{m_g}\left\{p_i\left[\begin{smallmatrix}u\\0\\ 
\tau\end{smallmatrix}\right]^{\T}+\left[\begin{smallmatrix}u\\0\\ 
\tau\end{smallmatrix}\right]p_i^{\T}:\, \left[\begin{smallmatrix}\tau\\
u\end{smallmatrix}\right]\in L_{k+1}\right\}.
\end{multline}
Next, using Lemma \ref{kw**} to take duals in \eqref{cone intersect}, 
we find 
\begin{multline}\label{untere}
{\mathscr FC}_+(D_{\Lambda})\supseteq\Sym{k+r+1}_+ +
\cone\left\{\left[\begin{smallmatrix}-\I&&\\&0_r&\\&&1\end{smallmatrix}
\right]\right\}\\
\quad+\sum_{i=0}^{m_g}\left\{p_i z^{\T}+zp_i^{\T}:\,
z\in{\mathscr H}(D_{\Lambda})^*\right\}+K_{{\mathscr H}}^*.
\end{multline}
Substituting \eqref{obige} into \eqref{untere}, exploiting the 
fact that 
\begin{equation*}
\left\{\left[\begin{smallmatrix}u\\0\\ 
\tau\end{smallmatrix}\right]\left[\begin{smallmatrix}v\\0\\ 
\sigma\end{smallmatrix}\right]^{\T}+\left[\begin{smallmatrix}v\\0\\ 
\sigma\end{smallmatrix}\right]\left[\begin{smallmatrix}u\\0\\ 
\tau\end{smallmatrix}\right]^{\T}:\,\left[\begin{smallmatrix}\tau\\u
\end{smallmatrix}\right], \left[\begin{smallmatrix}\sigma\\v
\end{smallmatrix}\right]\in L_{k+1}\right\}\subset
\Sym{k+r+1}_+ +\cone\left\{\left[\begin{smallmatrix}-\I&&\\&0_r&\\&&1
\end{smallmatrix}\right]\right\},
\end{equation*}
which follows from \eqref{for later use}, and using 
\begin{equation*}
\cone\left\{p_i p_i^{\T}+p_i p_i^{\T}:\,i=1,\dots,m_g\right\}
\subset\Sym{k+r+1}_+,
\end{equation*}
the inclusion claimed in the theorem is seen to hold true. Furthermore, 
it follows from \eqref{for later use too} that the inclusion is an equality 
when $m_g=0$. The fact that this is also true for $m_g=1$ follows 
from Sturm-Zhang \cite{sturm}, Theorem 3. 
\end{proof}

\begin{corollary}\label{approximation corollary}
Let $\Cov =U^{\T}U$ be the Cholesky factorization of $\Cov $. 
Then the following are sufficient conditions for \eqref{curly M} to 
hold: $\exists,\eta,\xi_{ij}\geq 0$, $(i\neq j=0,\dots,m_g)$, $s\in\R$ 
and $[\begin{smallmatrix}\tau_i\\u_i\end{smallmatrix}]\in 
L_{k+1}$, $(i=0,\dots,m_g)$ such that 
\begin{equation}\label{erschti}
\begin{bmatrix}\frac{1}{\sqrt{2}}&\frac{1}{\sqrt{2}}&\\
\frac{1}{\sqrt{2}}&-\frac{1}{\sqrt{2}}\\&&U\end{bmatrix}
\begin{bmatrix}\Half\\s\\x\end{bmatrix}\in L_{n+2}
\end{equation}
and 
\begin{multline}\label{lmi}
\begin{bmatrix}0&-\Half M^{\T}H&\Half M^{\T}(x-x_p)\\
-\Half H^{\T}M&\lambda H^{\T}\Cov H&H^{\T}(\Cov x_p
-\Half\overline{\mu})\\
\Half(x-x_p)^{\T}M&(\Cov x_p-\overline{\mu})^{\T}H&
\gamma-\lambda s+\overline{\mu}^{\T}(x-x_p)+\lambda x_p^{\T}
\Cov x_p\end{bmatrix}\\
-\eta\left[\begin{smallmatrix}-\I&&\\&0_r&\\&&1\end{smallmatrix}\right]
-\sum_{i\neq j=0}^{m_g}\xi_{ij}\left(p_i p_j^{\T}+p_j p_i^{\T}\right)\\
+\sum_{i=0}^{m_g}\left(p_i\left[\begin{smallmatrix}u_i\\0\\ \tau_i
\end{smallmatrix}\right]^{\T}+\left[\begin{smallmatrix}u_i\\0\\ \tau_i
\end{smallmatrix}\right]p_i^{\T}\right)\succeq 0.
\end{multline}
Furthermore, for $m_g\in\{0,1\}$ the above conditions are both 
necessary and sufficient for \eqref{curly M} to hold.
\end{corollary}

\begin{proof}
With ${\mathscr M}_{x,\gamma}$ defined as in Section \ref{sec:ellipsoidal}, 
condition \eqref{curly M} is of course the same as 
${\mathscr M}_{x,\gamma}\in\mathscr{FC}_+(D)$, see Lemma 
\ref{lem:sturm}. Proposition \ref{prop:subspace} shows that this is further 
equivalent to $\Lambda({\mathscr M}_{x,\gamma})\in
\mathscr{FC}_+(D_{\Lambda})$. By Theorem 
\ref{improved approximation theorem}, for this latter condition 
to hold it is sufficient to demand that 
\begin{multline}\label{prefinal}
\Lambda({\mathscr M}_{x,\gamma})\in
\Sym{k+r+1}_+ 
+\cone\left\{\left[\begin{smallmatrix}-\I&&\\&0_r&\\&&1
\end{smallmatrix}\right]\right\}\\
+\cone\left\{p_ip_j^{\T}+p_jp_i^{\T}:\,i\neq j\in\{0,\dots,m_g\}\right\}\\
+\sum_{i=0}^{m_g}\left\{p_i
\left[\begin{smallmatrix}u\\0\\ \tau\end{smallmatrix}\right]^{\T}
+\left[\begin{smallmatrix}u\\0\\ \tau\end{smallmatrix}\right]p_i^{\T}:\,
\left[\begin{smallmatrix}u\\ \tau\end{smallmatrix}\right]\in L_{k+1}\right\}.
\end{multline}
Further, since $\lambda\geq 0$, \eqref{prefinal} is equivalent to the 
existence of a $s\geq x^{\T}\Cov x$ such that 
\begin{multline*}
\Lambda({\mathscr M}_{x,\gamma})+\lambda\left[
\begin{smallmatrix}0&&\\&0&\\
&&x^{\T}\Cov x-s\end{smallmatrix}\right]\in
\Sym{k+r+1}_+ 
+\cone\left\{\left[\begin{smallmatrix}-\I&&\\&0_r&\\&&1
\end{smallmatrix}\right]\right\}\\
+\cone\left\{p_ip_j^{\T}+p_jp_i^{\T}:\,i\neq j\in\{0,\dots,m_g\}\right\}\\
+\sum_{i=0}^{m_g}\left\{p_i
\left[\begin{smallmatrix}u\\0\\ \tau\end{smallmatrix}\right]^{\T}
+\left[\begin{smallmatrix}u\\0\\ \tau\end{smallmatrix}\right]p_i^{\T}:\,
\left[\begin{smallmatrix}u\\ \tau\end{smallmatrix}\right]\in L_{k+1}\right\}
\end{multline*}
which is the same as \eqref{lmi} for some $\eta,\xi_{ij}\geq 0$, 
$(i\neq j=0,\dots,m_g)$, and $[\begin{smallmatrix}\tau_i\\u_i
\end{smallmatrix}]\in L_{1+k}$, $(i=0,\dots,m_g)$. Finally, the 
condition $s\geq x^{\T}\Cov x$ is easily seen to be equivalent to 
\eqref{erschti}. The last claim holds because by 
Theorem \ref{improved approximation theorem}, \eqref{prefinal} is 
equivalent to $\Lambda({\mathscr M}_{x,\gamma})\in
\mathscr{FC}_+(D_{\Lambda})$ when $m_g\in\{0,1\}$. 
\end{proof}

\section{Appendix A: Proof of Lemma \ref{lem:simple}}

\begin{proof}
Let $p_1,p_2 \in\cU$ and $\alpha \in [0,1]$. Let us first assume that 
$z^*(\alpha p_1+(1-\alpha)p_2)<+\infty$. For $\varepsilon>0$ there exists 
$x\in\cX$ such that 
\begin{align*}
z^*(\alpha p_1 +(1-\alpha)p_2)-\varepsilon&
\leq f(x,\alpha p_1 +(1-\alpha)p_2)\\
&\leq\alpha f(x,p_1)+ (1-\alpha)f(x,p_2)\\
&\leq\alpha z^*(p_1) + (1-\alpha)z^*(p_2).
\end{align*}
Since this is true for all $\varepsilon$, we find 
\begin{equation}\label{convexity}
z^*(\alpha p_1 +(1-\alpha)p_2)\leq\alpha z^*(p_1) + (1-\alpha)z^*(p_2).
\end{equation}
Now assume $z^*(\alpha p_1+(1-\alpha)p_2)=+\infty$. Then there exists 
a sequence $(x_n)_{\N}\subset\cX$ such that 
\begin{equation}\label{divergence}
f(x_n,\alpha p_1+(1-\alpha)p_2)\stackrel{n\rightarrow\infty}{\rightarrow}
+\infty.
\end{equation}
By the same argument as above, we have 
\begin{equation*}
f(x_n,\alpha p_1 +(1-\alpha)p_2)
\leq\alpha z^*(p_1) + (1-\alpha)z^*(p_2),
\end{equation*}
whence \eqref{divergence} establishes that at least one of 
$z^*(p_i)$ $(i=1,2)$ equals $+\infty$, so that \eqref{convexity} 
holds once again. 
\end{proof}

\section{Appendix B: Convexification of Problem (\ref{eq:Sharpe})} 

We use the notation introduced in Section \ref{subsubsection:Sharpe}. 
Using the fact that any $x\in\cX$ was assumed to satisfy the budget constraint 
$e^{\T}x$, Problem \eqref{eq:Sharpe} has an equivalent formulation 
\begin{align}
\max_{x\in\R^n}\,&f(x)=\frac{(\mu-r e)^{\T}x}{\sqrt{x^{\T}\Cov x}}
\label{first reformulation}\\
\text{s.t.}\quad&x\in\cX.\nonumber
\end{align}
The objective function $f(x)$ of this formulation is homogeneous of degree $0$ in $x$. 
Consider also the normalized problem, 
\begin{align}
\max_{y\in\R^{n}}\,&g(y)=(\mu-r e)^{\T}y\label{second reformulation}\\
\text{s.t.}\quad&y\in\R_+\cX,\nonumber\\
&y^{\T}\Cov y=1,\nonumber
\end{align}
where $\R_+\cX=\{\tau x:\,\tau\geq 0,\,x\in\cX\}$ is a tractable cone. For example, if $\cX$ 
is a polyhedron, as it is in most applications, $\R_+\cX$ is a polyhedral cone. 

Since $e^{\T}x=1$ for all $x\in\cX$ and $\Cov\succ 0$, we have $x^{\T}\Cov x>0$ for all $x\in\cX$. 
Therefore, any feasible solution $x$ of \eqref{first reformulation} provides 
$y=(x^{\T}\Cov x)^{-1/2}x$ as a feasible solution of \eqref{second reformulation}, and furthermore, 
$g(y)=f(x)$. Conversely, since any feasible $y$ of \eqref{second reformulation} satisfies $y\neq 0$, 
the vector $x=(e^{\T}y)^{-1}y\in\cX$ is feasible for \eqref{first reformulation} and satisfies 
$f(x)=g(y)$. This shows that \eqref{first reformulation} and \eqref{second reformulation} are 
equivalent: Instead of solving \eqref{first reformulation}, we may solve \eqref{second reformulation} 
and then construct an optimal solution $x^*=(e^{\T}y^*)^{-1}y^*$ of the first 
problem from an optimal solution $y^*$ of the second. 

Formulation \eqref{second reformulation} is furthermore equivalent to its relaxation  
\begin{align}
\max_{y\in\R^{n}}\,&g(y)=(\mu-r e)^{\T}y\label{third reformulation}\\
\text{s.t.}\quad&y\in\R_+\cX,\nonumber\\
&y^{\T}\Cov y\leq 1,\nonumber
\end{align}
since a feasible solution of \eqref{third reformulation} cannot be optimal unless 
$y^{\T}\Cov y=1$, assuming that there exist feasible $y$ for which $g(y)>0$ 
(if this is not the case, it is not rational to invest at all). Thus, the optimal solution $y^*$ of 
\eqref{second reformulation} may be found by solving the tractable convex problem 
\eqref{third reformulation}. 


\begin{thebibliography}{29}

\bibitem{Adom}
D.\ Adom. 
Robust Deviation Optimisation in Portfolio Theory. 
{\em M. Sc. Thesis}, 
Mathematical Institute, 
University of Oxford, September 2007. 

\bibitem{bertsimas}
D.\ Bertsimas, D.B.\ Brown and C.\ Caramanis. 
Theory and applications of robust optimization. 
{\em SIAM Rev.}.
Vol.\ 53, no.\ 3, pp.\ 464--501, 2010. 

\bibitem{bertsimasSim}
D.\ Bertsimas and M.\ Sim. 
The price of robustness. 
{\em Oper.\ Res.}. 
Vol.\ 52, no.\ 1, pp.\ 35--53, 2004. 

\bibitem{BTN1}
A.\ Ben-Tal and A.\ Nemirovski.
Robust convex optimization. 
{\em Mathematics of Operations Research}. 
Vol. 23, no. 4, pp. 769--805, 1998.

\bibitem{BTN2}
A.\ Ben-Tal and A. Nemirovski. 
Robust solutions of uncertain linear programs,
{\em Operations Research Letters}. 
Vol. 25, no. 1, pp. 1--13, 1999.

\bibitem{Ceria_Stubbs}
S.\ Ceria and R.\ Stubbs.
Incorporating estimation errors into portfolio selection: robust portfolio construction. 
{\em Axioma Research Paper}, no.\ 003, May 2006. 

\bibitem{cuckerPena}
F.\ Cucker and J.F.\ Pe\~na. 
A primal-dual algorithm for solving polyhedral conic systems with a finite-precision machine. 
{\em SIAM Journal on Optimization}.
Vol.\ 12, pp. 522--554, 2002.

\bibitem{ElGhaoui_Oks_Oustry} 
L.\ El Ghaoui, M.\ Oks, and F.\ Lebret.
Worst-case value at risk and robust portfolio optimization: A conic programming approach. 
{\em Operations Research}. 
Vol.\ 51, no.\ 4, pp. 543--556, 2003.

\bibitem{Goldfarb_Iyengar}
D.\ Goldfarb and G.\ Iyengar. 
Robust Portfolio Selection Problems. 
{\em Mathematics of Operations Research}, 
Vol. 28, no. 1, pp. 1--38, 2003.

\bibitem{gregory}
C.\ Gregory, K.\ Darby-Dowmann and G.\ Mitra. 
Robust optimization and portfolio selection: The cost of robustness. 
{\em European Journal of Operations Research}. 
Vol.\ 212, pp.\ 417--426.
 
\bibitem{Halldorsson_Tutuncu}
B.V.\ Halld\'{o}rsson and R.H.\ T\"{u}t\"{u}nc\"{u}.
An interior-point method for a class of saddle point problems.
{\em Journal of Optimization Theory and Applications}. 
Vol.\ 116, no.\ 3, pp.\ 559--590, 2003.

\bibitem{hauserS-Lemma}
R.A.\ Hauser. 
A new approach to Yakubovich's S-Lemma. 
{\em Technical Report} 1096, Numerical Analysis Group, 
University of Oxford. 

\bibitem{KY}
P.\ Kouvelis and G.\ Yu. 
Robust discrete optimization and its applications.
{\em Kluwer Academic Publishers}, Dordrecht, The Netherlands, 1997. 

\bibitem{krishnamurthy}
V.\ Krishnamurthy. 
Relative Robust Optimization. 
{\em MSC Thesis}, Carnegie Mellon University. 

\bibitem{Markowitz}
H.M.\ Markowitz. 
Portfolio Selection.
{\em Journal of Finance}. 
Vol.\ 7, pp.\ 77--91, 1952. 

\bibitem{murty-kabadi}
K.G.\ Murty and S.N.\ Kabadi. 
Some NP-complete problems in quadratic and nonlinear programming. 
{\em Math.\ Programming}.
Vol.\ 39, no.\ 2, pp.\ 117–-129, 1987. 

\bibitem{polik}
I.\ Polik and T.\ Terlaky. 
A survey of the S-Lemma. 
{\em SIAM Review}. 
Vol.\ 49, pp.\ 371--418, 2007.

\bibitem{rendl}
F.\ Rendl and H.\ Wolkowicz. 
A semidefinite framework for trust region subproblems with applications to large scale minimization.  
{\em Mathematical Programming}. 
Vol.\ 77, no.\ 2, pp.\ 273--299, 1997.

\bibitem{Schottle} 
K.\ Sch\"ottle, R.\ Werner and R.\ Zagst. 
Comparison and robustification of Bayes and Black-Litterman models. 
{\em Math.\ Methods of Oper.\ Res.}. 
Vol.\ 71, no.\ 3, pp.\ 453--475, 2010. 

\bibitem{Schottle2}
K.\ Sch\"ottle and R.\ Werner. 
Robustness properties of mean-variance portfolios. 
{\em Optimization}. 
Vol.\ 58, no.\ 6, pp.\ 641--663, 2009. 

\bibitem{capm} 
W.F.\ Sharpe. 
Capital asset prices: A theory of market equilibrium under conditions of risk. 
{\em Journal of Finance}. 
Vol.\ 19, no.\ 3, pp.\ 425-–442, 1964.

\bibitem{Sharpe} 
W.F.\ Sharpe.
The Sharpe Ratio. 
{\em Journal of Portfolio Management}. 
Vol.\ 21, no.\ 1, pp.\ 49--59, 1994. 

\bibitem{sedumi}
J.\ Sturm and I.\ Polik. 
SeDuMi 1.05 R5 user's guide. 
{\em http://sedumi.ie.lehigh.edu}. 

\bibitem{sturm}
J.\ Sturm and S.\ Zhang. 
On cones of nonnegative quadratic functions. 
{\em Math.\ Oper.\ Res.}. 
Vol.\ 28, no.\ 2, pp.\ 246--267, 2003. 

\bibitem{Taguchi} 
S.\ Taguchi. 
An Approximation Methods for Robust Deviation Decision Problems. 
{\em M. Sc. Thesis}, 
Department of Mathematical and Computing Sciences, 
Tokyo Institute of Technology, February 2005. 

\bibitem{Taguchi2}
A.\ Takeda, S.\ Taguchi, T.\ Tanaka. 
A Relaxation Algorithm with Probabilistic Guarantee for Robust Deviation Optimization Problems. 
{\em Comput.\ Optim.\ Appl.}. 
Vol.\ 47, pp.\ 1--31, 2010.

\bibitem{sdtp3} 
M.J.\ Todd, K.C.\ Toh and R.H.\ T\"ut\"unc\"u. 
On the implementation and usage of SDPT3 – a Matlab software package for semideﬁnite-quadratic-linear programming, version 4.0. 
{\em http://www.math.nus.edu.sg/~mattohkc/sdpt3.html}. 

\bibitem{Koenig_Tutuncu}
R.H.\ T\"{u}t\"{u}nc\"{u} and M.\ Koenig.
\newblock Robust Asset Allocation.
{\em Annals of Operations Research},
Vol. 132, pp. 157--187, 2004.

\bibitem{yakubovich}
V.A.\ Yakubovich. 
S-Procedure in nonlinear control theory. 
{\em Vestnik Leningrad University}. 
Vol.\ 1, pp.\ 62--77, 1971.


\end{thebibliography}
\end{document}